\documentclass[11pt]{article}
\usepackage{amssymb,amsmath,amsthm,mathtools,wasysym,calc,verbatim,enumitem,tikz,pgfplots,hyperref,url,mathrsfs,fullpage,bbm,comment, times, float}
\usepackage{tikz-cd}
\usetikzlibrary{arrows.meta,calc,positioning,decorations.pathreplacing}
\definecolor{dirblue}{RGB}{36,91,146}
\definecolor{dirorange}{RGB}{183,93,26}
\definecolor{dirgreen}{RGB}{31,119,91}
\definecolor{dirpurple}{RGB}{117,74,151}
\usepackage{booktabs}
\usepackage[noadjust]{cite}
\usepackage{comment,fullpage}
\usepackage[normalem]{ulem}
\mathtoolsset{showonlyrefs}
\pgfplotsset{compat=1.18}
\definecolor{ink}{HTML}{26333D}
\definecolor{muted}{HTML}{63717C}
\definecolor{teal}{HTML}{247D87}
\definecolor{ochre}{HTML}{B56D28}
\definecolor{plum}{HTML}{86618C}
\definecolor{hairline}{HTML}{B7C0C6}

\tikzset{
  every picture/.style={x=1cm,y=1cm,draw=ink,text=ink,
    line cap=round,line join=round,font=\small},
  point/.style={circle,fill=ink,draw=none,inner sep=1.6pt},
  edge/.style={line width=.75pt},
  accent/.style={line width=1.15pt},
  guide/.style={draw=muted,line width=.6pt,-{Stealth[length=2mm]}},
  note/.style={font=\footnotesize,text=muted,align=center},
  figcaption/.style={font=\small,text=ink,align=left,
    anchor=north west,inner sep=0pt},
  title/.style={font=\large\bfseries,anchor=west}
}

\newtheorem{theorem}{Theorem}
\newtheorem{definition}[theorem]{Definition}

\newtheorem{lemma}[theorem]{Lemma}
\newtheorem{claim}[theorem]{Claim}

\newtheorem{conj}[theorem]{Conjecture}
\newtheorem{fact}[theorem]{Fact}

\newtheorem*{claim*}{Claim}

\theoremstyle{remark}
\newtheorem*{remark*}{Remark}
\newtheorem{remark}[theorem]{Remark}

\usepackage{thmtools, thm-restate}

\numberwithin{theorem}{section}

\renewcommand{\phi}{\varphi}

\renewcommand{\leq}{\le}
\renewcommand{\geq}{\ge}

\newcommand{\cG}{\mathcal G}

\newcommand{\cC}{\mathcal C}

\usepackage{cancel}

\newcommand{\cT}{\mathcal T}

\newcommand{\E}{\mathbb{E}}

\def\1{\mathbbm{1}}

\newcommand{\diag}{\operatorname{diag}}

\newcommand{\supp}{\operatorname{supp}}

\def\rate{\mathrm{rate}}

\usepackage{tikz-cd}
\usetikzlibrary{positioning,arrows.meta,calc}

\renewcommand{\le}{\leqslant}
\renewcommand{\ge}{\geqslant}
\renewcommand{\P}{\mathbb{P}}

\newcommand{\PP}{\mathbb P}

\DeclareMathOperator{\PGL}{PGL}
\DeclareMathOperator{\GL}{GL}
\DeclareMathOperator{\SL}{SL}
\DeclareMathOperator{\val}{val}

\newcommand{\x}{\times}

\newcommand{\FF}{\mathbb{F}}

\usepackage{thm-restate}
\usepackage{array}



\newcommand{\F}{\mathbb{F}}
\newcommand{\cO}{\mathcal{O}}

\newcommand{\wt}{\operatorname{wt}}
\newcommand{\type}{\operatorname{type}}

\newcommand{\PRS}{\operatorname{PRS}}
\newcommand{\RS}{\operatorname{RS}}

\title{Good Quantum Locally Testable Codes from Product Expansion}
\author{
Mitali Bafna\textsuperscript{1}\thanks{mitalib@cs.washington.edu}
\qquad
Anqi Li\textsuperscript{1}\thanks{aqli@cs.washington.edu}
\qquad
Quynh T. Nguyen\textsuperscript{2}\thanks{thequynhqb@pm.me}
\\[0.5em]
{\normalsize
\textsuperscript{1}University of Washington
\qquad
\textsuperscript{2}UC Berkeley
}
}

\date{\vspace{-5ex}}

\begin{document}

\maketitle

\begin{abstract}
We construct quantum locally testable codes (LTCs) with constant rate, distance, soundness and locality under a variant of a product expansion conjecture of Bafna and Vyas~\cite{BafnaV} about Reed-Solomon codes. In particular, we use the high-dimensional expansion framework of Dinur, Lin and Vidick~\cite{DLV} for constructing quantum LTCs, instantiated with the non-Abelian cubical complexes of~\cite{RSV19}. Our code is obtained by equipping the complex with carefully chosen Reed-Solomon local codes whose symmetries are compatible with those of the complex. 
\end{abstract}

\paragraph{Statement of AI Use:} 
Motivated by the new conjectural construction of product expanding Reed-Solomon codes, the authors prompted GPT 6 to construct quantum LTCs using the framework of~\cite{DLV}, along with the non-Abelian cubical complexes considered in~\cite{HLMRZ}, assuming the product expansion conjecture of~\cite{BafnaV}. The model provided a construction that uses a variant of the conjecture of~\cite{BafnaV}. We spent significant effort in understanding the draft, extracting out the main ideas for our exposition and tracing many of the ideas to their origin in the literature. Upon further prompts, GPT 6 also provided a proof of the product expansion conjecture, but we did not verify its correctness, and hence have not included it in our manuscript. We are happy to share it with anyone who is interested.

%\tableofcontents

\section{Introduction}
Locally testable codes are error correcting codes that have a local tester that can probabilistically check whether a word is in the code using only few queries.

\begin{definition}[Locally testable codes]
For $\kappa>0$ and $q\in\mathbb N$, a linear code
$C\subseteq\mathbb F^n$ is \emph{$\kappa$-locally testable with $q$ queries}
if it is defined by a collection of parity checks, each involving at most
$q$ coordinates, such that every codeword is accepted with probability 1 and every word $f\in\mathbb F^n$ at relative distance $\Delta$ from $C$, violates at least
a $\kappa\,\Delta$-fraction of the checks.
\end{definition}
These originated in the celebrated works on program checking~\cite{BLR} and probabilistically checkable proofs~\cite{BFLS,ALMSS, AS}, in the form of low-degree testing of polynomials. Goldreich and Sudan~\cite{GoldreichS02} initiated a systematic study of LTCs, following which a long line of work~\cite{GoldreichS02,
Ben-SassonSVW03,Ben-SassonGHSV06,Ben-SassonS08pcp,Dinur07,
KoppartyMRS17,GopiKORS18} yielded constructions with \emph{nearly optimal} parameters. Although, achieving `good' LTCs, those with constant rate, relative distance, and locality simultaneously, remained open.
 
In a parallel vein, quantum low-density parity-check (qLDPC) codes
with good rate and distance have long been sought for their potential in performing efficient and fault-tolerant quantum computation~\cite{Gottesman14}.
A long line of work~\cite{TillichZ14,EvraKZ20,HastingsHO21,
PanteleevK22Distance,BreuckmannE21} developed constructions with
increasingly strong parameters.

The seminal works of~\cite{DinurELLM22} and~\cite{PanteleevK22} resolved the questions of constructing good LTCs and good quantum LDPC codes respectively, using a novel method based on high-dimensional expansion. These build upon the expander codes of~\cite{SipserS94}, that constructed good classical LDPC codes by defining a codespace on the edges of an expander graph by adding parity checks using good \textbf{`local'} codes at every vertex (found by brute force). They showed how the distance of the local codes translates to distance of the global code using graph expansion. These new works instead used a \emph{cubical complex}, which can be thought of as a graph that contains not just edges, but also squares (or 2-dimensional faces). They define a global code on the squares, by imposing local codes on the edges in a way so that the local view at a vertex is a tensor codeword. Here they used special local codes so that the resulting tensor codes has not only have good rate and distance, but  also is `robustly' locally testable. They then  lifted this local property to local testability of the full code building upon the local-to-global theorems of~\cite{KaufmanKazLub, EvraKaufman}. We will refer to this property of the local codes as \textbf{product expansion}, as coined by~\cite{PanteleevK22, KalachevP25} (see Definitions~\ref{def:prod-expansion-intro} and~\ref{def:product-expansion}).

This brings us to quantum LTCs that were introduced by Aharonov and Eldar~\cite{AE15} in connection with the quantum PCP conjecture~\cite{AharonovN02, aharonov2013guest}, formally defined as, 

\begin{definition}[Quantum locally testable CSS codes]
A CSS code $\mathcal Q=\operatorname{CSS}(C_X,C_Z)$ is specified by
two linear codes $C_X,C_Z\subseteq\mathbb F^n$ satisfying
$C_X^\perp\subseteq C_Z$. The parity checks defining $C_X$ and $C_Z$
give its $X$- and $Z$-type checks, respectively.
For $\kappa>0$ and $q\in\mathbb N$, we call the code $\mathcal Q$
\emph{$\kappa$-locally testable with $q$ queries}
if both $C_X$ and $C_Z$ are $\kappa$-locally testable with $q$ queries.
\end{definition}

Constructions of qLTCs with progressively improved parameters were obtained in~\cite{hastings2016quantum, leverrier2022towards, cross2024quantum}, and the currently best known parameters are by Dinur, Lin, and Vidick~\cite{DLV} who constructed almost-good qLTCs, achieving constant rate and locality with inverse-polylogarithmic relative distance and soundness. Their construction uses the product expansion of \emph{random codes} proved by~\cite{KalachevP25}. These almost-good qLTCs have been used in low-overhead fault-tolerant quantum computation~\cite{nguyen2025quantum}. Analogous to classical LTCs, qLTCs have implications for the quantum PCP conjecture and related problems: they  imply the NLTS theorem~\cite{EH17, ABN23}, quantum interactive oracle proofs~\cite{SunV26} and quantum multiparty computation as well as a `gap amplification' procedure for quantum CSPs in recent work by two of the authors~\cite{BNZ}. In all of these applications, good qLTCs yield better parameters than almost-good qLTCs.

\subsection{Our Work}
Motivated by the quantum PCP conjecture, a recent work by~\cite{BafnaV} constructed `private' PCPs and enroute to that gave a construction of improved quantum codes supporting multiplication. This result relies on a conjecture about the product expansion of carefully chosen Reed-Solomon (RS) codes (see Section~\ref{sec:pe-intro}), that if true would be the first structured codes apart from random codes that would satisfy this property. They used this conjecture to construct tensor product RS codes (see Corollary 4.5 therein) which give rise to quantum codes with near-linear rate and distance, sparse $X$ checks giving local testability, but having dense $Z$ checks. 

A natural question that arises is whether one can use these structured RS codes as local codes instead in the high-dimensional framework of~\cite{DLV} to obtain better parameters. Indeed, doing so we are able to build good qLTCs under a variant of the product expansion conjecture of~\cite{BafnaV},
\begin{theorem}[Informal (see Theorem~\ref{cor:main})]\label{thm:main}
Assuming Conjecture~\ref{conj:binary-pe}, there exists an explicit family of quantum locally testable CSS codes over a binary alphabet with constant rate, quantum distance, soundness and locality.
\end{theorem}

We cross the polylog-bottleneck in~\cite{DLV} by using \emph{non-Abelian} cubical complexes instead of Abelian ones. Using these introduces the difficulty that the local views are tensor codewords only if the symmetries of the local code play well with the complex used, in particular, one can no longer use random codes. This fact was already observed by~\cite{KalachevP25}, where they suggested that one might be able to use the non-Abelian complexes of Rungtanapirom--Stix--Vdovina~\cite{RSV19} (RSV), along with local RS codes to build quantum LTCs. It is also very similar to the approach that is discussed in~\cite{DinurELLM22}, which they did not carry out for their final result. Our main work goes into showing that an appropriate choice of RS codes that satisfies our product expansion conjecture ensures that every local view in the RSV complex is a tensor codeword. The DLV framework then applies automatically, and one can lift the product expansion of the local codes, to local testability of the global quantum code defined on the complex.

\paragraph{Future Directions:} Beyond codes, the HDX framework has led to many other applications in TCS such as sampling algorithms~\cite{AnariLG20}, property testing~\cite{BLM24, DDL24}, PCPs~\cite{BafnaMV25} and the construction of the first explicit lossless vertex expanders using the RSV complexes~\cite{HLMRZ}, which served as an inspiration for this work. We believe our work opens up exciting avenues for further applications. We use certain properties of algebraic HDX constructions (see Lemma~\ref{lem:construction}, property (4)) that to our knowledge have not yet seen applications. An immediate open question is be to obtain better qLTC and qLDPC codes that support multiplication~\cite{Dinur0Z25, GolowichLin}, and potentially lead to advances in classical and quantum PCP constructions.

\section{Technical Overview}\label{sec:overview}
We start by explaining the key ideas in the framework of~\cite{DinurELLM22} for constructing LTCs from cubical complexes\footnote{These ideas are quite similar to the quantum LDPC codes construction of~\cite{PanteleevK22}, but we choose to focus on classical LTCs.}. We then discuss the abstract framework of~\cite{DLV} that obtains quantum LTCs from 4-dimensional cubical complexes by generalizing their ideas, albeit runs into a `polylog-bottleneck'. We then discuss how to overcome this bottleneck using codes with symmetries that are `compatible' with the symmetries of the complex used. Finally we explain  our construction based on the RSV complexes and Reed-Solomon codes as local codes over them, as well as some intuition for their conjectured product expansion. %We show how a careful choice of Reed-Solomon codes as local codes over the complex leads to  all local views in the complex being tensor products of these codes. Additionally they satisfy our product expansion conjecture, which together with the local tensor view property, allows us to apply the local-to-global framework of DLV in a blackbox way to obtain good quantum LTCs.

%To illustrate the main ideas in our construction, it suffices to focus on constructing classical LTCs from 2-dimensional \emph{non-Abelian} cubical complexes with \emph{right-right} multiplication. The non-trivial part in the construction is to choose the inner codes so that their automorphisms are compatible with the complex and they are product expanding. The main work goes into showing that each local view at a vertex is a tensor codeword. While this is immediate in prior constructions, for us it is highly reliant on the properties of the complexes, as well as their compatibility with the codes we use (see Section~\ref{sec:}). Once we have this, the framework of~\cite{DLV} kicks in and give us the required local testability of the global quantum  code\footnote{Their framework is a local-to-global theorem prove that the nice properties of local views,  `coboundary or product expansion' give nice properties of the global complex -- `(co)systolic expansion'\MB{Check}. One can then define a natural quantum code over this complex using the set of local codes, that one sets up over this complex, a high-dimensional analogue of the idea in the expander codes of~\cite{SipserS94}.}. Hence our task boils down to showing why such nice local views arise for a particular choice of a complex and product-expanding codes on it.

\subsection{LTCs from Left-Right Cayley Complexes}
We start by giving an exposition of the classical LTCs of~\cite{DinurELLM22} that used the \emph{Left-Right} Cayley complexes.

\begin{definition}[2D-Left-Right Cayley Complex]\label{def:group}
Let $G$ be a finite group and let $A,B\subseteq G$ be nonempty sets closed under inverses. The vertex set is $V=G\times\{0,1\}^2$ and for $g\in G$ and $i,j\in\{0,1\}$, an $A$-edge joins $(g,i,j)$ to $(ga,i,1-j)$ for each $a\in A$, i.e. it acts by right multiplication and a $B$-edge joins it to $(bg,1-i,j)$ for each $b\in B$, i.e. acts by left multiplication. Define the set of squares by
\begin{equation}\label{eq:squareset}
\mathcal Q= \left\{
\begin{tikzcd}[column sep=huge,row sep=large]
(bg,10) \arrow[r,"a"] & (bga,11)\\
(g,00) \arrow[r,"a"'] \arrow[u,"b"] & (ga,01) \arrow[u,"b"'],
\end{tikzcd}
 \mid g\in G,\ a\in A,\ b\in B\right\}.
\end{equation}
\end{definition}
%\left\{\{(g,00),(ga,01),(bg,10),(bga,11)\}:

\paragraph{Defining the Code on the Complex:}
Fix orderings $A=(a_1,\ldots,a_m)$ and $B=(b_1,\ldots,b_n)$, a finite field $\E$, and let linear codes $C_A\subseteq \E^A$, $C_B\subseteq \E^B$ with the specified ordering over $A, B$. A word $z\in \E^{\mathcal Q}$ assigns one field element $z(Q)$ to every square. We then define the code $\cC$ by demanding that the squares around every $A$-edge $((g,00),(ga,01))$ or $((g,10), (ga,11))$, obtained by varying $b$ in the specified order, lie in $C_B$ and similarly the squares around any $B$-edge lie in $C_A$:
\begin{equation}\label{eq:abstractcode}
\mathcal C=\left\{z\in \E^{\mathcal Q}:\begin{array}{ll}
(z(Q_e(b_1)),\ldots,z(Q_e(b_n)))\in C_B&\text{for every }A\text{-edge }e,\\
(z(Q_e(a_1)),\ldots,z(Q_e(a_m)))\in C_A&\text{for every }B\text{-edge }e
\end{array}\right\},
\end{equation}
where $Q_e(\cdot)$ denotes the square obtained on the edge $e$ by moving via the specified element.

\paragraph{Local Views at Vertices:}
Take any codeword $z\in\mathcal C$ and fix a vertex $v$. Let $Q_v(a,b)\in\mathcal Q$ be the unique square containing the $A$-edge and $B$-edge incident to $v$ with labels $a,b$ at $v$. 
 The squares around $v$ can be arranged into an $m\times n$ matrix $Z_v$ with $Z_v(a,b):=z(Q_v(a,b)), a\in A, b\in B$, which we will refer to as the \emph{local view} of $v$. The row labeled by $a$ contains exactly the symbols on squares containing the $A$-edge $e$ incident to $v$ with label $a$ and the same holds for columns labeled by $B$. Thus the word $z$ restricted to this local view belongs to $C_A \otimes C_B$.

\paragraph{Deriving the rate and local testability of $\cC$ for an appropriate choice of $C_A, C_B$:} 
A simple constraint counting argument implies that the rate of $\cC$ equals $2\rate(C_A)+2\rate(C_B)-3$. Thus to get non-zero rate for $\cC$ via counting, the sum of their rates must be larger than $3/2$.

Secondly, consider the set of parity checks given by picking a vertex, and testing whether the squares around it belong to the code $C_A\otimes C_B$. We have shown above that all codewords satisfy this property. A tensor code is locally testable via the natural row and column checks, but if it is additionally `robustly' testable via these checks, or equivalently has product expansion, the works~\cite{DinurELLM22,PanteleevK22} show that one can lift this property to the full code, i.e. $\cC$ itself is locally testable via these vertex-based parity checks! We discuss product expansion later on, but for now it suffices to think of this as an abstract property we will want our local codes to satisfy in order to build LTCs from this framework.

\subsection{Generalizing to Higher Dimensions}
The work of~\cite{DLV} (DLV) generalized this framework to show that `4-dimensional' Cayley complexes yield quantum LTCs, in a similar  way by defining local codes on faces. Although there seems to be no way of generalizing Left-Right non-Abelian Cayley complexes to higher dimensions, they obtained higher dimensional complexes by moving to Abelian Cayley complexes, where we have 4 generating sets, and each of them acts by right multiplication. At a high level, they require the following properties from a 4-dimensional Cayley complex,
\begin{enumerate}
\item The underlying Cayley graphs in each direction should be a good spectral expander, i.e. the normalized adjacency matrix should have second largest eigenvalue at most some small constant.
\item A 4-tuple of codes $(C_i)_{i\in 4}$, which will serve as our constant-sized local codes. We require the rates of these codes to be close to $1$ and that they satisfy 4-dimensional product expansion\footnote{Technically, one pair needs to have rates close to 1, and the other needs to have rates close to 0, and we also need $(C_i^\perp)$ to satisfy product expansion.}.
\item The property that defining a code $\cC$ by adding constraints on the (local views of the co-dimensional 1 faces of the) complex according to the codes $C_i$ (or their duals) satisfies that the local view of every face in the complex is a tensor codeword; this is the higher-dimensional analogue of the local view at a vertex being the tensor codeword $C_A \otimes C_B$.
\end{enumerate}

%It turns out that the latter complexes are non-trivial to construct even in 2 dimensions, since one needs to choose the generating sets $A,B$ carefully so that the complex has squares at all, unlike the Left-Right case, where every choice yields squares. Secondly, as we explain later, the inner codes need to be `compatible' with the complex, otherwise the local views are not tensor codewords.

\paragraph{Expansion Bottleneck:} It is a well-known fact that Abelian Cayley graphs on $n$ vertices need $\Theta(\log n)$ generators to achieve constant spectral gap. This parameter is essentially responsible for the $\log n$-losses in the soundness/query-complexity of the quantum LTCs of DLV. 

\subsection{Constructing LTCs over Right-Right Non-Abelian Cayley Complexes}
Thus to construct good quantum LTCs we need to move to non-Abelian Cayley complexes under \emph{right} multiplication, which also has a natural analogue in higher dimensions (see Definition~\ref{def:coset}). The main difficulty is already present in using these complexes in 2 dimensions for constructing good classical LTCs, thus we only focus on this for the overview.

\begin{definition}[2D-Right-Right-Cayley Complex]\label{def:rr-cayley}
Again we have a set of vertices $V = G \times \{0,1\}^2$ and two generating sets $A,B$ that are closed under inverse but unlike the Left-Right Cayley complex, both the set of edges act by right multiplication: an $A$-edge joins $(g,i,j)$ to $(ga,i,1-j)$ for each $a\in A$ and a $B$-edge joins it to $(gb,1-i,j)$ for each $b\in B$. 

To get a well-defined set of squares, we additionally require that $AB=BA$ and $|AB|=|A||B|$. This implies that, for each $a\in A,b\in B$, there are unique $b'\in B,a'\in A$ with $ab'=ba'$, which can be represented by the permutations $\beta_a(b)=b'$ and $\beta_b(a)=a'$, so that we have the set of squares,
\begin{equation}
\mathcal Q= \left\{
\begin{tikzcd}[column sep=huge,row sep=large]
(gb,10) \arrow[r,"\beta_b(a)"] & (ga\beta_a(b),11)\\
(g,00) \arrow[r,"a"'] \arrow[u,"b"] & (ga,01) \arrow[u,"\beta_a(b)"'],
\end{tikzcd} \mid g\in G,\ a\in A,\ b\in B\right\}.
\end{equation}
\end{definition}
% \[
% \begin{tikzcd}[column sep=huge,row sep=large]
% (gb,10) \arrow[r,"\beta_b(a)"] & (ga\beta_a(b),11)\\
% (g,00) \arrow[r,"a"'] \arrow[u,"b"] & (ga,01) \arrow[u,"\beta_a(b)"'],
% \end{tikzcd}
% \]
% giving us the set of squares,
% \begin{equation}\label{eq:squareset}
% \mathcal Q=
% \left\{\{(g,00),(ga,01),(gb,10),(ga\beta_a(b),11)\}:
%  g\in G,\ a\in A,\ b\in B\right\}.
% \end{equation}
% \end{definition}

\paragraph{The global code:} Defining a code $\cC$ on this complex is more subtle than in the Left-Right case -- given a fixed edge, say $((g,00), (ga,01))$ enumerating the squares touching it using $(g,00)$ in the fixed order over $B$, is not the same as enumerating them from $(ga,01)$. Thus, for an $A$-edge $e$, we choose the endpoint whose second bit is $0$ and enumerate the squares in the prescribed order $b_1,\ldots,b_n$; denote the resulting squares by $Q_e(b_1),\ldots,Q_e(b_n)$. For a $B$-edge $e$, we enumerate them from the endpoint whose first bit is $0$, using the $A$-edges in the order $a_1,\ldots,a_m$.

\textbf{Local Views at a Vertex under Symmetry Assumptions on the Local Codes:}
Again we fix $z\in\mathcal C$ and a vertex $v$, and look at its local view given by the matrix, 
\[
Z_v(a,b):=z(Q_v(a,b)), \qquad a\in A, b\in B.
\]
One can check that if $v$ has type $00$, then just like the case of the Left-Right Cayley construction, each row of $Z_v$ belongs to $C_B$ and column belongs to $C_A$, i.e. $Z_v \in C_A \otimes C_B$. But consider a vertex $v$ of type $01$. Each column of $Z_v$ will belong to $C_A$, but one can check that each row $a$ will be the permuted codeword,
\[
u\circ\beta_{a}
=\bigl(u_{\beta_{a}(b_1)},\ldots,u_{\beta_{a}(b_n)}\bigr),
\]
for some $u\in C_B$. The analogous property holds for vertices of types $10$ and $11$.

If our codes $C_A, C_B$ are invariant under these permutations, i.e. $u\circ \beta_a \in C_B$ for all $u\in C_B$, and $u\circ \beta_b \in C_A$ for all $u\in C_A$, then we will get that each vertex view belongs to $C_A \otimes C_B$.

%It belongs to $C_B$ by the symmetry assumption~\eqref{eq:symmetry}.

%Similarly, column $k$ lists the symbols around the $B$-edge $e$ incident to $v$ with label $b_k$. Writing $u=u_e(z)\in C_A$, this column is $u$ if $v=w_e$, and $u\circ\beta_{b_k}$ otherwise. Hence every column belongs to $C_A$. The tensor code $C_A\otimes C_B$ is precisely the set of matrices with columns in $C_A$ and rows in $C_B$. Therefore
% \begin{equation}\label{eq:abstracttensor}
% z\in\mathcal C\quad\Longrightarrow\quad Z_v\in C_A\otimes C_B
% \quad\text{at every vertex }v.
% \end{equation}

% Write a word in $E^n$ as $u=(u_{b_1},\ldots,u_{b_n})$, and define
% \[
% u\circ\beta_a:=\bigl(u_{\beta_a(b_1)},\ldots,u_{\beta_a(b_n)}\bigr).
% \]
% For $w=(w_{a_1},\ldots,w_{a_m})$, define $w\circ\beta_b$ analogously. Assume
% \begin{equation}\label{eq:symmetry}
% u\in C_B\ \Longrightarrow\ u\circ\beta_a\in C_B\quad(a\in A),\qquad
% w\in C_A\ \Longrightarrow\ w\circ\beta_b\in C_A\quad(b\in B).
% \end{equation}
% Thus these permutations preserve the codes with their fixed coordinate orderings.

\subsection{Our Construction in 2-Dimensions: Complex and Local Codes}
Thus our task reduces to choose local codes that are invariant under the maps $\beta$ of the complex. Our construction builds on the ideas suggested in~\cite{DinurELLM22}, in their abandoned approach, and~\cite[footnote~36]{PK24}. They noticed that projective Reed-Solomon codes are natural choices for local codes, along with the complexes of \cite{RSV19} (or other similar algebraic complexes). In more detail, the vertices of the RSV complexes is $G=\operatorname{PGL}_2(\mathbb F_{2^r})$, the group of invertible matrices over an extension field $\E = \F_{2^r}$\footnote{Precisely, the complex from~\cite{RSV19} turns out to have the vertex set being cosets of some $G$ with respect to the subgroup $\operatorname{PGL}_2(\mathbb{F}_2)$, but we ignore this distinction for the overview.}. As in the constructions of Ramanujan
graphs~\cite{lubotzky1988ramanujan} (and
complexes~\cite{LSV1}), the $A$-edge (analogously for $B$-edges) around each vertex can be labeled by the set of \emph{projective lines} or elements of $\P^1(\F_{s_A}) = \F_{s_A}\cup \infty$, where $s_A$ is some power of $2$. This set can alternatively be viewed as a set of 2-dimensional vectors: $a \to (a, 1)^T$ for $a\in \F_{s_A}$ and $a \to (1, 0)^T$ if  $a = \infty$, which we refer to as $r_a$. For each fixed $b$, each $\beta_b$ permutation over the labels of $A$ can then be implemented using a $2\times 2$ invertible matrix $B_b$: $B_b r_a=\lambda_a r_{\beta_b(a)}$ for a nonzero scalar $\lambda_a$. This symmetry property plays well with projective RS (PRS) codes, which is the set of homogenous bivariate degree $d$ polynomials with coefficients in $\E$, evaluated over $\P^1(\F)$. For any such polynomial $p$, we have that, 
\[p \circ \beta_b(r_a) = p(r_{\beta_b(a)}) = p(\lambda_a^{-1}B_b r_a) = \lambda_a^{-d} q(r_a),\]
for some degree $d$ polynomial $q = p\circ B_b$. Note that if the scaling factor in front of $q$ did not exist, we would get that $p\circ \beta_b$ belongs to the code (and the same holds for the $\beta_a$ maps), and then by the argument above each vertex view would be a tensor PRS codeword. 

That is, the main difficulty is how to handle these scaling factors. We deal with this issue using the fact that one can set up certain matrices on faces of the complex carefully, such that they are `compatible' with the $B_a, B_b$ matrices (see property (4) in Lemma~\ref{lem:construction}). Then, using these matrices, we choose a scaling factor for each edge and define the constraint on it as a scaled PRS codes (thus each edge has a different constraint). Finally, we show that one can choose appropriate scalars per $(v,Q)$ pair, so that the scaled code satisfies the fact that each local view around a vertex $v$ belongs to a tensor PRS code (with coefficients in $\E$) evaluated over the domain $\P^1(\F_{s_A})\times \P^1(\F_{s_B})$. We defer a more detailed discussion for the intuition here to Section~\ref{sec:complex-and-code}.

\subsubsection{Product Expansion of Reed-Solomon Codes}\label{sec:pe-intro}
The second property we need from our pair of local codes (which in our case are PRS codes evaluated over $\P^1(\F_s)$), is that they satisfy product expansion (PE) in the high-rate regime. Let us define 2d-PE (see Definition~\ref{def:product-expansion} for the high-dimensional analogue): let $C_1\subseteq\E^{S_1}$ and $C_2\subseteq\E^{S_2}$ be linear codes, where $n_i=|S_i|$, and view a word on $S_1\times S_2$ as a matrix. For
a matrix $M$, let $\supp(M)$ be the set of its nonzero entries, and let $|M|_{\rm col}$ and $|M|_{\rm row}$ denote the number of its nonzero
columns and rows, respectively.

\begin{definition}[Two-dimensional product expansion]
\label{def:prod-expansion-intro}
The pair $(C_1,C_2)$ is $\rho$-product-expanding if every matrix $M$
that is the sum of a matrix whose columns belong to $C_1$ and a matrix
whose rows belong to $C_2$ admits such a decomposition $M=M_1+M_2$
satisfying
\[
|\supp(M)|
\geq \rho\bigl(n_1|M_1|_{\rm col}+n_2|M_2|_{\rm row}\bigr).\footnote{This two-dimensional notion is equivalent to robust testability of the tensor code $C_1\otimes C_2$. In higher dimensions (see Definition~\ref{def:product-expansion}) product expansion is a stronger condition.}
\]
\end{definition}

Polishchuk and Spielman~\cite{polishchuk1994nearly} proved that any pair of Reed-Solomon codes evaluated over $\F^\star_q$ is product expanding with $\rho = \Omega(1)$ when the \emph{sum} of their rates is \emph{less than one}, which is a tight condition, and thus these codes cannot be used as local codes for LTCs. 

To see why their sum-of-rates condition is tight, consider a pair of RS codes $C_1,C_2$ evaluated on $S_1,S_2$, respectively. Intuitively, the main obstruction to product expansion without such a rate restriction is a low-degree bivariate curve $F(X,Y)$ containing too many points of $S_1\times S_2$. For example, when $S_1=S_2=\mathbb F_q^\star$, the  diagonal $X-Y=0$ contains $q-1$ grid points, which is much more than its degree that is a constant. This example is the reason that PE fails when the sum of the rates becomes $1$.

%To see why their result is tight, consider a pair of RS codes $C_1, C_2$ evaluated over the domains $S_1, S_2$ respectively. Intuitively,  when these have arbitrarily large rates, PE boils down to showing that there is no low degree curve $F(X,Y)=0$ that intersects $S_1\times S_2$ on too many points, relative to the degree of $F$. Thus when $S_i = \F^\star_q$, the curve $X-Y=0$ intersects this grid in $q-1$ points, which is an unbounded number despite the constant degree of $F$, and an obstruction to PE.

The work of~\cite{BafnaV} avoids this issue by choosing $S_1, S_2$ carefully: they show that if the $S_i$ are \emph{multiplicative subgroups of coprime orders} of $\F^\star_q$, for prime $q$, not only does the above obstruction disappear (since $|S_1\cap S_2|=1$), but the pair has constant PE in the high-rate regime (with an analogous conjecture in higher dimensions). We need such a PE result for PRS codes; fortunately, one can show that a PRS code  evaluated over $\P^1(\F_{s})$ is equivalent to an RS code (of similar degree) evaluated over a multiplicative subgroup of $\F^\star_{s^2}$ of size $s+1$ (see Lemma~\ref{lem:prs-rs-pe}). We want our ambient field to contain $\F^\star_{s^2}$, which means we cannot take it to be of prime order. Thus we move to extension fields, and use the RSV complexes over $\E = \F_{2^r}$. One could hope that the same conjecture also holds for arbitrary coprime order multiplicative subgroups $S_i$ of $\E^\star$. That turns out to be false due to low-degree curves that intersect $S_1\times S_2$ in too many points. But restricting to a pair of coprime order subgroups whose  order is of the form $2^{m}+1$ seems to avoid such obstructions. 

In conclusion, we conjecture this modified PE statement (see Conjecture~\ref{conj:binary-pe}) in high-dimensions and choose the 4-dimensional complex so that in each direction, the edge labels lie in $\P^1(\F_{s_i})$ where $s_i$ is a power of $2$. The local views at a vertex are then shown to be tensor products of PRS codes evaluated over $\P^1(\F_{s_i})$, which satisfy PE under our conjecture. This gives us all the properties required to get good qLTCs using DLV's framework.

\section{Preliminaries}\label{sec:prelims}

\subsection{Cubical Complexes}
Our complex will be obtained from the construction of~\cite{RSV19}. Instead of the Cayley complexes defined above, we will get a complex whose set of vertices contains the cosets of a group $G$ with respect to a subgroup $H$, just like in the case of Schreier graphs that generalize Cayley graphs (the setting of $H = \{1\}$).

\begin{definition}[A $k$-dimensional complex on cosets]\label{def:coset}
Let $H$ be a subgroup of a finite group $G$. A coset is a set $gH=\{gh:h\in H\}$, and $G/H$ denotes the set of these cosets. Let $D_1,\ldots,D_k\subseteq G$ be nonempty subsets satisfy $D_i^{-1}=D_i$ and $HD_iH = D_i$. We write $D/H=\{dH:d\in D\}$, let $n_i=|D_i/H|$ and assume
\begin{equation}\label{eq:coset}
 D_iD_j=D_jD_i \quad(i\ne j),
 \qquad
 |D_S/H|=\prod_{i\in S}n_i
 \quad(\varnothing\ne S\subseteq[k]),
\end{equation}
where $D_S=\prod_{i\in S}D_i$.
The vertices of the complex are given by $(gH,\mathbf b)\in(G/H)\times\{0,1\}^k$.
The direction-$i$ edges are
\[
 \{(gH,\mathbf b),(gdH,\mathbf b+\mathbf e_i)\}
 \qquad(d\in D_i),
\]
where $\mathbf e_i$ is the $i$th unit vector and addition is modulo two.
For each nonempty $S\subseteq[k]$, the cubes in directions $S$
are the $2^{|S|}$-sized vertex sets of the form,
\[
 \left\{
 \left(g_T H,\mathbf b+\sum_{i\in T}\mathbf e_i\right)
 :T\subseteq S
 \right\} 
\]
where $\mathbf b\in\{0,1\}^k$, $g_T\in G$, and
$g_T^{-1}g_{T\cup\{i\}}\in D_i$ for $ T \subseteq S$ and $i \in S \setminus T$.
Each such cube has dimension $|S|$, and its faces arise
by fixing some of the type bits in $S$.
\end{definition}

Note that the neighbor sets are independent of the representative $g$ of $gH$, because $hD_i=D_i$ for every $h\in H$ and $i \in [k]$. Each vertex has $n_i$ incident edges in direction $i$.  The size condition in~\eqref{eq:coset} says that each two-edge path in two distinct directions $i$ then $j$ is uniquely determined by its endpoints, and we also have a unique path in the reversed order. That is, by considering $gd_iH\to gH\to gd_jH$ and its reverse, we can identify fourth corner of the unique square containing two given edges at $gH$.

\begin{definition}
   Let $X$ be a cubical complex with vertex
type classes $V_{\type}$, indexed by $\type\in\{0,1\}^4$.
Every vertex has $n_i\ge1$ neighbors in direction $i$,
and direction-$i$ edges flip only the $i$th type bit.
For $\type_i=0$, we write $X_{i,\type}$ for the
$n_i$-regular bipartite graph consisting of these edges
between $V_{\type}$ and $V_{\type+\mathbf e_i}$,
where $\mathbf e_i$ is the $i$th unit vector and addition
is modulo two.

For $0\le\lambda<1$, we say that $X$ is
\emph{$\lambda$-expanding} if
$\lambda_2(X_{i,\type})\le\lambda$ for every direction $i$
and every $\type$ with $\type_i=0$.
\end{definition}

\subsection{Projective Reed-Solomon Codes}\label{sec:prs}

We define projective Reed--Solomon codes, identify their duals, and prove their symmetry under M\"obius transformations.

\paragraph{Projective coordinates.}
Let $s$ be a prime power and let $\E$ be a finite field containing $\F_s$.
The projective line $\PP^1(\F_s)$ is the set of one-dimensional subspaces of $\F_s^2$. Write $\langle r\rangle$ for the line spanned by a nonzero vector $r$.
We identify its $s+1$ elements with $\F_s\cup\{\infty\}$ using the representatives
\[
r_x=(x,1)^T\quad(x\in\F_s),\qquad r_\infty=(1,0)^T.
\]
Write $\GL_2(\F_s)$ for the invertible $2\times2$ matrices, $\SL_2(\F_s)$ for its determinant-one subgroup, and $\PGL_2(\F_s)$ by identifying two matrices which differ by a scalar multiple.

\begin{definition}[Projective Reed--Solomon code]\label{def:prs}
For $0\le d\le s-1$, define
\[
\PRS_d(\F_s;\E)=\left\{(f(r_x))_{x\in\PP^1(\F_s)}:
f(X,Y)=\sum_{j=0}^d f_jX^jY^{d-j},\quad f_j\in \E\right\}.
\]
We abbreviate this code by $\PRS_d$ when the fields are clear.
\end{definition}

Thus finite coordinates are $f(x,1)$, and the coordinate at infinity is $f(1,0)=f_d$.
It is well-known that the code has parameters $[s+1,d+1,s+1-d]_{\E}$.

\begin{claim}[Dual of projective Reed-Solomon codes]\label{lem:prs-dual}
For $0\le d\le s-1$, we have
$
\PRS_d(\F_s;\E)^\perp=\PRS_{s-1-d}(\F_s;\E).
$
\end{claim}

\begin{proof}
Let $f,g$ be homogeneous polynomials of degrees $d,s-1-d$.
Their product has degree $s-1$, so the finite-field power sums give
\[
\sum_{x\in\F_s}f(x,1)g(x,1)=-f(1,0)g(1,0).
\]
Thus their evaluation vectors are orthogonal, including the
coordinate at infinity. The two codes have dimensions $d+1$ and $s-d$, whose sum is $s+1$, proving equality.
\end{proof}

\paragraph{M\"obius transformations.}
For $B=\begin{pmatrix}a&b\\c&e\end{pmatrix}\in\GL_2(\F_s)$, let $\beta_B$ be the permutation of projective coordinates induced by $B$,
\begin{equation}\label{eq:prs-mobius-scalars}
\beta_B(x)=\langle Br_x\rangle,\qquad
Br_x=\lambda_B(x)r_{\beta_B(x)}.
\end{equation}
In scalar notation, $\beta_B(x)=(ax+b)/(cx+e)$ for finite $x$, with a zero denominator giving $\infty$; also $\beta_B(\infty)=a/c$ if $c\ne0$, and $\infty$ otherwise.
For finite $x$ with $cx+e\ne0$, the scalar is $\lambda_B(x)=cx+e$.

\begin{claim}[Mobius symmetry]\label{lem:prs-mobius}
For every $B\in\GL_2(\F_s)$,
\[
\left\{\bigl(\lambda_B(x)^d(u\circ\beta_B)_x\bigr)_x:
u\in\PRS_d\right\}=\PRS_d,
\]
where $(u\circ\beta_B)_x=u_{\beta_B(x)}$.
\end{claim}

\begin{proof}
For $u=(f(r_x))_x$, we have
$\lambda_B(x)^d(u\circ\beta_B)_x
=f(\lambda_B(x)r_{\beta_B(x)})=f(B r_x)=:(f\circ B)(r_x)$.
Since $B$ is invertible, $f\mapsto f\circ B$ is a bijection on
homogeneous degree-$d$ polynomials, proving the claim.
\end{proof}

\section{The Cubical Complex and Local Codes}
As discussed in the overview, the main contribution of our work is to set up local Reed-Solomon codes on the non-Abelian cubical complex of RSV, in a way such that all the local views are tensor codewords. We start by defining the properties of the complex that we will use, then describe how the local codes are assigned and finally show that all local views are tensor codewords.

\subsection{The Complex and Its Symmetries} \label{sec:complex-and-code}

The following lemma describes the cubical complex and its symmetries. The proof is deferred to Section~\ref{sec:complex}.

\begin{lemma}[Cubical Complex Construction]\label{lem:construction}
Fix positive integers $m_1,\ldots,m_4$ divisible by $8$, and set
$s_i=2^{m_i}$ for $i\in[4]$. For each such choice and arbitrarily
large integers $r$, there is a finite cubical complex
$X=\Gamma\backslash\cT$, where $\Gamma\leq\cG$ acts freely on vertices,
with the following properties. Set
$G=\PGL_2(\FF_{2^r})$ and $H=\PGL_2(\FF_2)\leq G$.
\begin{enumerate}
\item The vertices are $(G/H)\times\{0,1\}^4$.
There are subsets $D_i\subset G$ satisfying
$HD_i=D_iH=D_i$, $D_i^{-1}=D_i$, and $|D_i/H|=s_i+1$.
Writing $\mathbf e_i$ for the $i$th unit vector, the direction-$i$
edges are
\begin{equation}\label{eq:coset-edges}
 (gH,\vec{b})\ \sim_i\ (gdH,\vec{b}+\mathbf e_i),
 \qquad dH\in D_i/H.
\end{equation}
Here addition is modulo two.
At each vertex, edges in distinct directions determine a unique cube.
Squares and higher cubes are the images of product cubes in $\cT$.
Moreover, $D_iD_j=D_jD_i$ for $i\ne j$.

\item At every vertex, the $s_i+1$ incident direction-$i$ edges
are labeled by $A_i=\PP^1(\F_{s_i})$.
The type of a vertex $v=(gH,\vec{b})$ is $\type(v)=\vec{b}$.

\item Fix nonzero representatives $r_a^{(j)}\in \F_{s_j}^2$ for
$a\in A_j$. For each directed direction-$i$ edge $e:v\to w$
and each $j\ne i$, let $\beta_e^{(j)}:A_j\to A_j$ pair the
direction-$j$ edges opposite each other across $e$.
There is a matrix $B_e^{(j)}\in\SL_2(\F_{s_j})$ such that
\begin{equation}\label{eq:label-transport}
 \beta_e^{(j)}(a)=\langle B_e^{(j)}r_a^{(j)}\rangle.
\end{equation}
% Reversing $e$ uses $(B_e^{(j)})^{-1}$ and the corresponding
% reciprocal scalars.
\item For every vertex $v$ of every $4$-cube $Q$ of $X$, one can choose matrices
$F_i(v,Q)\in\SL_2(\F_{s_i})$, for $1\leq i\leq4$, such that:
\begin{enumerate}
\item If $a_i$ labels the direction-$i$ edge of $Q$ at $v$, then
column $\type_i(v)+1$ of $F_i(v,Q)$ spans $a_i$.
\item For every direction-$i$ edge $e:v\to w$, oriented so that
$\type_i(v)=0$ and $\type_i(w)=1$, there are
$P_{v,e},P_{w,e}\in\SL_2(\F_{s_i})$, independent of $Q$, such that
every $Q\supset e$ satisfies
\begin{equation}\label{eq:compatibility}
\begin{aligned}
 F_i(v,Q)&=P_{v,e}\begin{pmatrix}\eta&*\\0&\eta^{-1}\end{pmatrix},
 &F_i(w,Q)&=P_{w,e}\begin{pmatrix}\eta&0\\ *&\eta^{-1}\end{pmatrix},\\[3pt]
 F_j(w,Q)&=B_e^{(j)}F_j(v,Q) &&(j\ne i),
\end{aligned}
\end{equation}
where $\eta \in \F_{s_i}^\times$ and the starred entries may depend
on $e,Q$, and $B_e^{(j)}$ is the matrix from item 3.
\end{enumerate}
\end{enumerate}
\end{lemma}

The symmetry properties of the complex are crucial to ensure that the local views of vertices are consistent, thereby yielding a well-defined global code. In the rest of this section, we will show how to set up local codes on the complex, and then prove that the local views of vertices are tensor codewords.

\subsection{Setting up Local Codes on the 2d Complex}

The main idea of our construction can be described in two dimensions, with the higher-dimensional case following essentially the same calculations. Thus, we will focus on the 2d version of the complex in Lemma~\ref{lem:construction}.

We associated each square with a scalar in a finite field $\E$. We equip this complex with a set of local codes, one for every edge in $X(2)$. The global code is defined as the set of all assignments of field elements to squares that satisfy the local constraints on edges. The edges' local codes are rescaled and permuted versions of two base projective Reed--Solomon codes, one for each direction.

\begin{definition}[The Base Codes]\label{def:base-codes}
Let $s_i=2^{m_i}$, where $m_i\ge1$ for $i\in\{1,2\}$, and let
$\E$ be a finite field containing both $\F_{s_1}$ and $\F_{s_2}$.
Fix an integer $h\ge0$ and, for each $i\in\{1,2\}$, an integer
$d_i$ satisfying
\[
0\le d_i\le s_i-1,
\qquad s_i-1\mid(2^h+1)d_i.
\]
Let $A_i=\PP^1(\F_{s_i})$, with a fixed coordinate order, and define
\[
C_i=\PRS_{d_i}(\F_{s_i};\E)\subseteq \E^{A_i},
\qquad i\in\{1,2\}.
\]
\end{definition}

We note that the divisibility condition on the degrees will be crucial.

Our goal is to design the edge local codes so that the local view around each vertex belongs to a tensor code. The main challenge for doing so is that, in our non-Abelian cubical complex, the two endpoints of an edge index its incident squares using their respective local labels that are not necessarily the same. Instead, the labels are related by the permutations $\beta_e$ as specified by Lemma~\ref{lem:construction}(3). If the base codes were invariant under the permutations
$\beta_e$, these differences in ordering would be harmless.
Projective Reed--Solomon codes instead have the \emph{scaled}
permutation symmetry of Claim~\ref{lem:prs-mobius}.
We therefore introduce rescaling weights, chosen so that
one rescaling of each square symbol works simultaneously
for its row and its column.  Making these weights compatible at both endpoints of an edge requires
an additional modification, which we resolve by introducing an edge-independent permutation.
Fortunately, rescaling and permutation do not affect the product expansion of tensor codes, which is a property concerning Hamming weights only.

We now describe the permutations and rescaling coefficients that will be used to define the local codes on edges.

\begin{fact}
On a finite field of characteristic two, the map $x\mapsto x^2$ is a field automorphism known as the Frobenius map.
\end{fact}

\paragraph{Edge-independent permutations.}
Let $P$ be the $h$-fold iterate of the Frobenius automorphism: $x\mapsto x^{2^h}$. This induces a permutation of the projective coordinates $A_i$ by permuting the field elements $\mathbb{F}_{s_i}$ and fixes $\infty$.
We extend this naturally to define the permuted codes $PC_i$. So, e.g., $P C_1$ consists of the words $(f(r_{a^{2^h}}))_{a \in A_1}$, with $f$ a polynomial of homogeneous degree $d_1$ (see Definition~\ref{def:prs}). 
The $\mathbb{F}_2$-linearity of the Frobenius map, $(x+y)^{2^h} = x^{2^h} + y^{2^h}$, will be crucial. Later, we will sometimes raise a vector or matrix to the $2^h$-th power entrywise, we will be explicit when doing so.

\paragraph{Vertex weights.} The rescaling coefficients are defined in terms of the matrices from Lemma~\ref{lem:construction}(4) and depend on the type of the vertex and the square.

First, each vertex associates a weight to each incident square. We use labels $a\in A_1$ and $b\in A_2$.
 For a vertex $v$ of type $t_1t_2$ and a square $Q=Q_v(a,b)$, define $c_1(v,Q)\in\F_{s_1}^{\times}$ and $c_2(v,Q)\in\F_{s_2}^{\times}$ by
\begin{equation}\label{eq:cdefinition}
\begin{aligned}
\text{column }t_1+1\text{ of }F_1(v,Q)&=c_1(v,Q)r_a,\\
\text{column }t_2+1\text{ of }F_2(v,Q)&=c_2(v,Q)r_b.
\end{aligned}
\end{equation}
These scalars can depend on the whole square $Q$, hence on both $a$ and $b$. Define
\begin{equation}\label{eq:nu}
\nu_v(Q)=
\begin{cases}
c_1(v,Q)^{d_1}c_2(v,Q)^{d_2},&v\text{ of type }00,\\[2pt]
c_1(v,Q)^{d_1}c_2(v,Q)^{2^h d_2},&v\text{ of type }01,\\[2pt]
c_1(v,Q)^{2^h d_1}c_2(v,Q)^{d_2},&v\text{ of type }10,\\[2pt]
c_1(v,Q)^{2^h d_1}c_2(v,Q)^{2^h d_2},&v\text{ of type }11.
\end{cases}
\end{equation}

\paragraph{Local codes.} We can now define the local codes on edges. Assign a symbol $z(Q)\in \E$ to every square $Q$. For a direction-$i$ edge $e$, let its \emph{designated endpoint} be the endpoint $w$ with $\mathrm{type}(w)_i=0$. The squares incident to $e$ are denoted by $Q_e(b)$ where $b$ is the label of the direction-$j$ edge ($j \neq i$) specified by $w$ that completes the square. The squares are ordered by the labels given by the designated endpoint. The local code enforced by edge $e$ is $C_j$, but rescaled and permuted according to the rules specified below. Let $D_e$ be the operation that multiplies the coordinate for each square $Q_e(b)$ by $\nu_w(Q_e(b))$. Write $D_eC=\{D_ec:c\in C\}$.

The following table specifies the local codes imposed on edges.
\begin{table}[htbp]
\centering
\begin{tabular}{@{}llll@{}}
\toprule
Direction & Endpoint types & Designated endpoint & Required edge code\\
\midrule
$1$ & $00\leftrightarrow10$ & $00$ & $D_eC_2$\\
$1$ & $01\leftrightarrow11$ & $01$ & $D_ePC_2$\\
$2$ & $00\leftrightarrow01$ & $00$ & $D_eC_1$\\
$2$ & $10\leftrightarrow11$ & $10$ & $D_ePC_1$\\
\bottomrule
\end{tabular}
\caption{Local codes imposed on edges. Each direction-$i$ edge has an designated endpoint of type $t_i=0$. The rescaling weights $D_e$ are specified by the designated endpoint. The permutaion $P$ is applied when the designated endpoint has type $t_j=1$ for the other direction $j\ne i$.}
\label{tab:edge-checks}
\end{table}

% \textbf{The global code.} Define
% \begin{equation}\label{eq:code}
% \mathcal C=\{z\in \E^{\mathcal Q}: z(Q_e) \text{ belongs to its listed edge code for every }e\}.
% \end{equation}
For example, on a direction-$1$ edge $e$ between types $0t_2$ and $1t_2$, the check at $e$ requires that there exists a homogeneous degree-$d_2$ polynomial $f$ satisfies
\begin{equation}\label{eq:edgepoly}
z(Q_e(b))=\nu_w(Q_e(b))f(r_{b^{2^{ht_2
}}})\qquad(b\in A_2),
\end{equation}
where $w$ is the type-$0t_2$ endpoint of $e$.

\begin{remark}\label{rem:weight-intuition} We give some intuition for choosing the weights $\nu_v(Q)$ in~\eqref{eq:nu} and the permutation $P$.
  
First consider the construction without the permutation $P$. Consider again an edge $e=(w,v)$ of type $00\leftrightarrow 10$ and a square $Q=Q_e(b)$.
The matrices in Lemma~\ref{lem:construction}(4) provide the
vectors $c_2(v,Q)r_{\beta_e(b)}$ and $c_2(w,Q)r_b$. Since a homogeneous
degree-$d$ polynomial satisfies $f(cr)=c^df(r)$ for any scalar $c$, using the matrix identity in Lemma~\ref{lem:construction}(4) gives
\[
f(r_{\beta_e(b)})
=\left(\frac{c_2(w,Q)}{c_2(v,Q)}\right)^{d_2}f(B_e r_b).
\]
This suggests a weight contribution of the form $c_2(v,Q)^{d_2}$ at every vertex $v$. Applying the same intuition to direction-$2$ edges suggests another weight contribution of the form $c_1(v,Q)^{d_1}$, and thus an overall weight of $\nu_v(Q)=c_1(v,Q)^{d_1}c_2(v,Q)^{d_2}$ at every vertex $v$. However, this does not work because the direction-$2$ contribution $\frac{c_1(w,Q)^{d_1}}{c_1(v,Q)^{d_1}}$ can depend on $Q$ and affect the direction-$1$ rescaling guarantee. We therefore want to modify the exponents so that the two direction rules agree.

The key observation is that, using the triangular forms of the $F$ matrices in Lemma~\ref{lem:construction}(4) we can infer $\frac{c_1(w,Q)^{d_1}}{c_1(v,Q)^{m d_1}} \propto (\text{some }Q\text{-dependent quantity})^{(m +1) d_1}$. The latter becomes independent of $Q$ if $(s_1-1) \mid (m+1)d_1$ by Fermat's little theorem. At vertex $v$, replacing an evaluation
label $x$ by $x^{m}$ changes the associated scalar factor
from $c_1(v,Q)^{d_1}$ to $c_1(v,Q)^{m d_1}$. It turns out that we also want the map $x \mapsto x^m$ to be a field automorphism,
% (to induce an edge label permutation and to be $\mathbb{F}_2$-linear) 
and the Frobenius map provides such a map. Thus applying the permutation $P$ at vertices where
$t_i=1$ together with rescaling weights
$\nu_v(Q)=\prod_{i=1}^2 c_i(v,Q)^{2^{h t_i}d_i}$ works.
\end{remark}

\subsection{Local Views of Vertices in the 2d Construction}\label{sec:local-view}

In this section, we show that the edge constraints from the preceding subsection enforce that the local view of every vertex is a tensor code, up to rescaling and permuting coordinates.

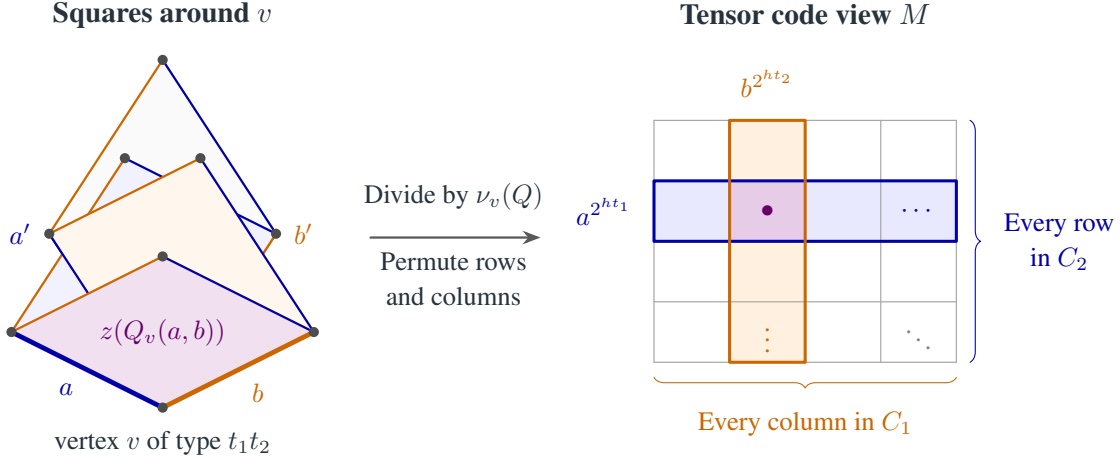
\begin{figure}
\centering

\[\begin{tikzpicture}[x=1cm,y=1cm,font=\small,
  line cap=round,line join=round,
  dirone/.style={draw=blue!65!black,line width=.85pt},
  dirtwo/.style={draw=orange!80!black,line width=.85pt},
  dot/.style={circle,fill=black!70,inner sep=1.35pt},
  flow/.style={-{Stealth[length=2.4mm]},draw=black!65,line width=.85pt}]

\node[font=\bfseries] at (2.5,6.4) {Squares around $v$};
\node[font=\bfseries] at (11,6.4) {Tensor code view $M$};

% Four squares of the product of two stars. The selected direction-1
% edge a and direction-2 edge b meet in the purple square Q_v(a,b).
\begin{scope}[shift={(2.5,1.2)}]
  \coordinate (v) at (0,0);
  \coordinate (a) at (-2,1);
  \coordinate (aa) at (-1.5,2.3);
  \coordinate (b) at (2,1);
  \coordinate (bb) at (1.5,2.3);
  \coordinate (aabb) at (0,4.6);
  \coordinate (abb) at (-.5,3.3);
  \coordinate (aab) at (.5,3.3);
  \coordinate (ab) at (0,2);

  \path[fill=black!2] (v)--(aa)--(aabb)--(bb)--cycle;
  \draw[dirone] (v)--(aa) (bb)--(aabb);
  \draw[dirtwo] (v)--(bb) (aa)--(aabb);

  \path[fill=blue!5] (v)--(a)--(abb)--(bb)--cycle;
  \draw[dirone] (v)--(a) (bb)--(abb);
  \draw[dirtwo] (v)--(bb) (a)--(abb);

  \path[fill=orange!7] (v)--(aa)--(aab)--(b)--cycle;
  \draw[dirone] (v)--(aa) (b)--(aab);
  \draw[dirtwo] (v)--(b) (aa)--(aab);

  \path[fill=violet!12] (v)--(a)--(ab)--(b)--cycle;
  \draw[dirone] (b)--(ab);
  \draw[dirtwo] (a)--(ab);
  \draw[dirone,line width=1.8pt] (v)--(a)
    node[midway,below left=1pt,text=blue!65!black] {$a$};
  \draw[dirtwo,line width=1.8pt] (v)--(b)
    node[midway,below right=1pt,text=orange!80!black] {$b$};
  \node[text=violet!80!black] at (0,1) {$z(Q_v(a,b))$};

  \foreach \point in {v,a,aa,b,bb,aabb,abb,aab,ab}
    \node[dot] at (\point) {};
  \node[left=3pt,text=blue!65!black] at (aa) {$a'$};
  \node[right=3pt,text=orange!80!black] at (bb) {$b'$};
  \node[below=5pt] at (v) {vertex $v$ of type $t_1t_2$};
\end{scope}

% One invertible change of the whole local view: entrywise rescaling,
% followed by the prescribed permutations of the coordinate labels.
\draw[flow] (5.25,3.45)--(7.5,3.45);
\node[align=center] at (6.35,4.0) {Divide by $\nu_v(Q)$};
\node[align=center] at (6.35,2.9) {Permute rows\\and columns};

% The matrix is schematic; the highlighted row and column represent
% arbitrary edge checks, not a code of length four.
\path[fill=blue!9] (9,3.4) rectangle (13,4.2);
\path[fill=orange!12] (10,1.8) rectangle (11,5);
\path[fill=violet!20] (10,3.4) rectangle (11,4.2);
\foreach \x in {9,10,11,12,13}
  \draw[black!35,line width=.45pt] (\x,1.8)--(\x,5);
\foreach \y in {1.8,2.6,3.4,4.2,5}
  \draw[black!35,line width=.45pt] (9,\y)--(13,\y);
\draw[dirone,line width=1.15pt] (9,3.4) rectangle (13,4.2);
\draw[dirtwo,line width=1.15pt] (10,1.8) rectangle (11,5);
\node[text=violet!80!black] at (10.5,3.8) {$\bullet$};
\node[text=blue!65!black] at (12.5,3.8) {$\cdots$};
\node[text=orange!80!black] at (10.5,2.2) {$\vdots$};
\node[text=black!50] at (12.5,2.2) {$\ddots$};
\node[anchor=east,text=blue!65!black] at (8.85,3.8)
  {$a^{2^{h t_1}}$};
\node[above=5pt,text=orange!80!black] at (10.5,5)
  {$b^{2^{h t_2}}$};

\draw[decorate,decoration={brace,amplitude=4pt},draw=blue!65!black]
  (13.18,5)--(13.18,1.8)
  node[midway,right=9pt,align=center,text=blue!65!black]
    {Every row\\in $C_2$};
\draw[decorate,decoration={brace,amplitude=4pt},draw=orange!80!black]
  (13,1.6)--(9,1.6)
  node[midway,below=9pt,text=orange!80!black]
    {Every column in $C_1$};

% \node[text=black!80!black] at (7.25,.15)
%   {$\displaystyle
%     M\bigl(a^{2^{h t_1}},b^{2^{h t_2}}\bigr)
%       =\frac{z(Q_v(a,b))}{\nu_v(Q_v(a,b))}$};
% \node at (7.25,-.85)
%   {$\text{All edge checks at }v
%      \quad\Longleftrightarrow\quad M\in C_1\otimes C_2$};
\end{tikzpicture}\]
    \caption{The local view of vertex $v$ is a tensor code up to rescaling and permutations.}
    \label{fig:placeholder}
\end{figure}

\begin{lemma}[Local views are Tensor Codes]\label{lem:local-tensor}
Fix a vertex $v$ of type $t_1t_2$.
An assignment $z$ to the squares containing $v$ satisfies
the checks in Table~\ref{tab:edge-checks} on every edge
incident to $v$ if and only if there exists
$M\in C_1\otimes C_2$ such that
\[
z(Q_v(a,b))
=\nu_v(Q_v(a,b))\,
 M\bigl(a^{2^{h t_1}},b^{2^{h t_2}}\bigr)
\qquad(a\in A_1,\ b\in A_2).
\]
\end{lemma}

\begin{proof}
We need to show that division by $\nu_v(Q)$ turns each row into a codeword in $P^{t_2}C_2$, and each column into a codeword in $P^{t_1}C_1$. See Figure~\ref{fig:placeholder}.

First consider the rows.
If $t_1=0$, the direction-$1$ edge checks are defined at $v$, so containment in $P^{t_2}C_2$ follows directly from Table~\ref{tab:edge-checks}. Suppose $t_1=1$, and consider a fixed edge $e = (w, v)$, so $w$ has type $0t_2$ and is the designated endpoint of $e$. Let $B_e^{(2)}$ be the matrix from Lemma~\ref{lem:construction}(3) and $\beta_e^{(2)}$ be the associated edge label permutation. Let us omit the direction superscript because there is only one direction transverse to $e$, i.e., we will write $B_e, \beta_e$. Consider a label $b \in A_2$ and the square completed by an edge labeled by $b$ at $w$, $Q=Q_e(b)$. 
Let $a_w,a_v$ be the labels of $e$ at $w,v$, respectively, so $Q=Q_v(a_v,\beta_e(b)) = Q_w(a_w,b)$.

The edge check at $e$ requires
\[
z(Q)=\nu_w(Q) f(r_{b^{2^{h t_2}}}),
\]
where $f$ is homogeneous of degree $d_2$ independent of $b$. 
% Our goal is to show that the rescaled row
% \begin{align}
%    \frac{\nu_w(Q)f(r_{b^{2^{h t_2}}}) }{\nu_v(Q)}, \qquad Q = Q_e(b), b\in A_2,  \label{eq:scaled-row}
% \end{align}
% lies in $P^{t_2}C_2$.

Since $w$ has type $0t_2$ and $v$ has type $1t_2$,
the weights according to ~\eqref{eq:cdefinition},~\eqref{eq:nu} are
\[
\nu_w(Q)=c_1(w,Q)^{d_1}c_2(w,Q)^{2^{h t_2}d_2},
\qquad
\nu_v(Q)=c_1(v,Q)^{2^h d_1}c_2(v,Q)^{2^{h t_2}d_2}. 
\]
So the rescaled entry is
\begin{align}
\frac{z(Q)}{\nu_v(Q)} = \frac{c_1(w,Q)^{d_1}}{c_1(v,Q)^{2^h d_1}}
\left(\frac{c_2(w,Q)}{c_2(v,Q)}\right)^{2^{h t_2}d_2}  f(r_{b^{2^{h t_2}}})
\label{eq:rescaled}
\end{align}

We will show the following two claims:

\begin{claim}[Direction-1 Contribution is $b$-independent]\label{claim:direction-1}  $\frac{c_1(w,Q)^{d_1}}{c_1(v,Q)^{2^h d_1}} = \alpha_{e}$ independent of $Q$.
  
\end{claim}

\begin{claim}[Direction-2 Contribution is a Permuted Codeword]\label{claim:direction-2} $\left(\frac{c_2(w,Q)}{c_2(v,Q)}\right)^{2^{h t_2}d_2}  f(r_{b^{2^{h t_2}}}) = f'( r_{\beta_e(b)^{2^{h t_2}}})$, where $f'$ is a homogeneous degree-$d_2$ polynomial.
\end{claim}

We first use the claims to conclude the lemma before proving the claims. 

Substituting the claims into~\eqref{eq:rescaled} gives
\[ 
\frac{z(Q)}{\nu_v(Q)} = \alpha_e f'( r_{\beta_e(b)^{2^{h t_2}}}) =: g(r_{\beta_e(b)^{2^{h t_2}}}),
\]
where $g$ is a homogeneous degree-$d_2$ polynomial.
Recall that $Q=Q_e(b)=Q_v(a_v,\beta_e(b))$, hence changing the variable name $\beta_e(b)\mapsto b$ gives
\[ 
\frac{z(Q_v(a_v,b))}{\nu_v(Q_v(a_v,b))} = g(r_{b^{2^{h t_2}}}).
\]
Therefore,
\[
M\bigl(a_v^{2^{h t_1}},b^{2^{h t_2}}\bigr)
=\frac{z(Q_v(a_v,b))}{\nu_v(Q_v(a_v,b))}
=g(r_{b^{2^{h t_2}}}).
\]
Since $b\mapsto b^{2^{h t_2}}$ permutes $A_2$, this row of $M$
belongs to $C_2$.

Applying the argument to every
direction-$1$ edge gives all row constraints. Exchanging
directions gives all column constraints. We conclude that $M\in C_1\otimes C_2$.

We now prove the claims.

\begin{proof}[Proof of Claim~\ref{claim:direction-1}] This proof is where our choice of permutation $P$ and the degree divisibility condition in Definition~\ref{def:base-codes} are used.

By~\eqref{eq:cdefinition} and Lemma~\ref{lem:construction}(4),
\[
\begin{aligned}
c_1(w,Q)r_{a_w}
&=F_1(w,Q)\begin{pmatrix}1\\0\end{pmatrix} =P_{w,e}
  \begin{pmatrix}\eta&*\\0&\eta^{-1}\end{pmatrix}
  \begin{pmatrix}1\\0\end{pmatrix}
 =\eta P_{w,e}\begin{pmatrix}1\\0\end{pmatrix},\\
c_1(v,Q)r_{a_v}
&=F_1(v,Q)\begin{pmatrix}0\\1\end{pmatrix} =P_{v,e}
  \begin{pmatrix}\eta&0\\ *&\eta^{-1}\end{pmatrix}
  \begin{pmatrix}0\\1\end{pmatrix}
 =\eta^{-1}P_{v,e}\begin{pmatrix}0\\1\end{pmatrix}.
\end{aligned}
\]
Rearranging gives
\[
P_{w,e}\begin{pmatrix}1\\0\end{pmatrix}
 =\eta^{-1}c_1(w,Q)r_{a_w},
\qquad
P_{v,e}\begin{pmatrix}0\\1\end{pmatrix}
 =\eta c_1(v,Q)r_{a_v}.
\]
Thus the first column of $P_{w,e}$ is a scalar multiple
of $r_{a_w}$, and the second column of $P_{v,e}$ is a
scalar multiple of $r_{a_v}$. Denote these scalars by
$p_w$ and $p_v$. They are nonzero because the matrices
are invertible. Moreover, the matrices $P_{w,e}$, $P_{v,e}$ and the chosen
representatives $r_{a_w},r_{a_v}$ depend only on the edge
and its endpoints, so $p_w,p_v$ are independent of $Q$
even though $\eta$ and $c_1(w,Q),c_1(v,Q)$ may vary with $Q$.
Consequently,
\[
c_1(w,Q)=p_w\eta,
\qquad
c_1(v,Q)=p_v\eta^{-1}.
\]
Since $s_1-1\mid(2^h+1)d_1$, we have $\eta^{(2^h+1)d_1}=1$ by Fermat's little theorem, and hence,
\begin{align}
  \frac{c_1(w,Q)^{d_1}}{c_1(v,Q)^{2^h d_1}}
=\frac{p_w^{d_1}}{p_v^{2^h d_1}}\eta^{(2^h+1)d_1}
=\frac{p_w^{d_1}}{p_v^{2^h d_1}} =: \alpha_e.
\end{align}
\end{proof}

\begin{proof}[Proof of Claim~\ref{claim:direction-2}] The idea of this claim is to use the scaled permutation symmetry of projective Reed-Solomon codes from Claim~\ref{lem:prs-mobius}.

By Lemma~\ref{lem:construction}(4),
\[
F_2(v,Q)=B_eF_2(w,Q).
\]
Both vertices $v,w$ have second type bit $t_2$. Comparing column
$t_2+1$ on both sides and using~\eqref{eq:cdefinition} gives
\[
c_2(v,Q)r_{\beta_e(b)}
=c_2(w,Q)B_er_b.
\]

Rearranging,
\[
\frac{c_2(w,Q)}{c_2(v,Q)}r_b = B_e^{-1}r_{\beta_e(b)}
.
\]
Because $(x+y)^{2^h} = x^{2^h} + y^{2^h}$ for any field elements $x,y$ in characteristic $2$, we can take entrywise $2^{h t_2}$-th powers both sides and get
\begin{align}
\left(\frac{c_2(w,Q)}{c_2(v,Q)}\right)^{2^{h t_2}} r_{b^{2^{h t_2}}}= (B_e^{-1})^{2^{h t_2}}r_{\beta_e(b)^{2^{h t_2}}}.
\end{align}
Hence,
\[
  \left(\frac{c_2(w,Q)}{c_2(v,Q)}\right)^{2^{h t_2}d_2}
  f(r_{b^{2^{h t_2}}}) = 
  f\!\left(
    \left(\frac{c_2(w,Q)}{c_2(v,Q)}\right)^{2^{h t_2}} r_{b^{2^{h t_2}}}
  \right) =
  f\!\left((B_e^{-1})^{(2^{h t_2})}r_{\beta_e(b)^{2^{h t_2}}}\right).
\]
Letting $f'(r_x) := f((B_e^{-1})^{(2^{h t_2})} r_x)$ proves the claim.
\end{proof}

\end{proof}

\subsection{Generalization to Higher Dimensions}
\label{sec:four-dimensional-local-codes}

We now extend the 2-dimensional construction in the previous subsection to the 4-dimensional complex of
Lemma~\ref{lem:construction} by placing local codes on $3$-cubes. This essentially a straightforward extension.

\begin{definition}[Four-dimensional local codes] \phantom{}
\label{def:four-dimensional-local-codes}
We assign a local code to each cube (3-face) in the 4-dimensional complex as follows:
\begin{enumerate}
\item \textbf{Base codes.}
Let $s_i=2^{m_i}$ with $m_i\ge1$, let $\E$ be a finite field
containing every $\F_{s_i}$, and fix an integer $h\ge0$. Choose
\[
0\le d_i\le s_i-1,\qquad s_i-1\mid(2^h+1)d_i,
\qquad i\in[4],
\]
and set $A_i=\PP^1(\F_{s_i})$ and
$C_i=\PRS_{d_i}(\F_{s_i};\E)\subseteq \E^{A_i}$.
As before, $P$ permutes the evaluation labels by $x\mapsto x^{2^h}$ while
fixing $\infty$.

\item \textbf{Vertex weights.}
For a vertex $v$ of type $t_1t_2t_3t_4$, write
$Q=Q_v(a_1,a_2,a_3,a_4)$ for the $4$-cube with these edge labels at $v$,
and define $c_i(v,Q)\in\F_{s_i}^{\times}$ and $\nu_v(Q)$ by
\begin{equation}\label{eq:nu-four-dimensional}
\begin{aligned}
\text{column }t_i+1\text{ of }F_i(v,Q)&=c_i(v,Q)r_{a_i}
\qquad(i\in[4]),\\
\nu_v(Q)&=\prod_{i=1}^4 c_i(v,Q)^{2^{h t_i}d_i}.
\end{aligned}
\end{equation}

\item \textbf{Local codes.}
Assign a symbol $z(Q)\in \E$ to every $4$-face (hypercube) $Q$.
For a $3$-face (cube) $\tau$ missing direction $i$, its \emph{designated vertex}
is the unique vertex $w\in\tau$ whose type bits in directions $j\ne i$
are zero. Write $Q_{\tau}(b)$ for the $4$-face obtained by adjoining
the direction-$i$ edge labeled $b\in A_i$ at $w$.
In this coordinate order, let $D_{\tau}$ multiply the coordinate for
$Q_{\tau}(b)$ by $\nu_w(Q_{\tau}(b))$.
The local code is $C_{\tau}=D_{\tau}P^{t_i}C_i$, where $t_i$ is the
constant $i$th type bit on $\tau$, as listed in
Table~\ref{tab:three-cube-checks}.
\begin{table}[htbp]
\centering
\begin{tabular}{@{}cc@{}}
\toprule
Type bit $t_i$ on $\tau$ & Required code $C_{\tau}$\\
\midrule
$0$ & $D_{\tau}C_i$\\
$1$ & $D_{\tau}PC_i$\\
\bottomrule
\end{tabular}
\caption{Local codes on $3$-faces missing direction $i$, for every $i\in[4]$.}
\label{tab:three-cube-checks}
\end{table}

Thus the constraint at $\tau$ requires a homogeneous degree-$d_i$
polynomial $f$ such that
\[
z(Q_{\tau}(b))=\nu_w(Q_{\tau}(b))f(r_{b^{2^{ht_i}}})
\qquad(b\in A_i).
\]
The global code consists of assignments satisfying these constraints
on every $3$-face.
\end{enumerate}
\end{definition}

The next lemma extends Lemma~\ref{lem:local-tensor}
to the local view of any face.

\begin{lemma}[Local Views are Tensor Codes]\label{lem:local-tensor-four-dimensional}
Let $f$ be any face with direction set $S\subseteq[4]$, and fix a vertex
$v\in f$ of type $t_1t_2t_3t_4$.
An assignment $z$ to the $4$-faces containing $f$ satisfies the constraints
$C_{\tau}$ for every $3$-face $\tau$ containing $f$ if and only if there exists
$M\in\bigotimes_{i\notin S}C_i$ such that
\[
z(Q)=\nu_v(Q)\,
M\bigl((a_i^{2^{h t_i}})_{i\notin S}\bigr),
\qquad Q=Q_v(a_1,a_2,a_3,a_4)\supseteq f.
\]
\end{lemma}

\begin{proof} 
Fix a $3$-face $\tau$ missing direction $i$. We first show that its
constraint, divided coordinatewise by $\nu_v(Q)$ and expressed in
the labels at any $v\in\tau$, is exactly $P^{t_i}C_i$. Let $w$ be the
designated vertex of $\tau$. At $w$, dividing the constraint coordinatewise
by $\nu_w(Q)$ gives $P^{t_i}C_i$ by definition.

For any other vertex $v\in\tau$, choose a path from $w$ to $v$
along edges of $\tau$. We show that the desired description is
equivalent at the two endpoints of each edge, so it holds at $v$.
Consider on this path an edge $e:u\to u'$ in $\tau$ of direction $k\ne i$,
oriented from type bit $t_k=0$ to $t_k=1$.
Claim~\ref{claim:direction-1}, with direction $1$ replaced by $k$,
shows that the direction-$k$ factor in $\nu_u(Q)/\nu_{u'}(Q)$
is independent of $Q$. For $j\notin\{i,k\}$, the direction-$j$
edge labels at $u,u'$ are fixed by $\tau$. Comparing the corresponding
columns in Lemma~\ref{lem:construction}(4) shows that
$c_j(u,Q)/c_j(u',Q)$ is also independent of $Q$.
The only remaining factor is in direction $i$, where the proof of
Claim~\ref{claim:direction-2} applies with matrix $B_e^{(i)}$.
As in Claim~\ref{claim:direction-2}, the invertible linear substitution and the nonzero constant factors
give an equivalence between the rescaled codeword symbols at
$u$ and $u'$. Applying this equivalence successively along the chosen path proves the assertion at $v$.

Now fix $f$ and $v$ as in the statement. The $4$-faces containing $f$
are indexed by the labels $(a_i)_{i\notin S}$ at $v$.
Divide their symbols by $\nu_v(Q)$ and reindex each coordinate by
$a_i\mapsto a_i^{2^{ht_i}}$, then by the preceding paragraph, the
constraints on the $3$-faces containing $f$ become exactly the
$C_i$ constraints on every direction-$i$ line for each $i\notin S$.
This concludes the proof.
\end{proof}

\section{Rate, Distance, and Soundness Analysis}
\label{sec:dlv-parameters}
In this section, we define the quantum code using local constraint systems on the cubes described in the previous section. We instantiate the parameters and then establish
its rate, distance, and soundness by adapting the corresponding
arguments of Dinur--Lin--Vidick~\cite{DLV}, under the product expansion conjecture.

\subsection{Defining the Quantum Code}
\label{subsec:cubical-construction}
We define a cochain complex and quantum code using the cubical complex
of Lemma~\ref{lem:construction} and the local codes from
Definition~\ref{def:four-dimensional-local-codes}.

Recall from Lemma~\ref{lem:construction} that the cubical complex $X$ has vertex set
\[
X(0)=(G/H)\times\{0,1\}^4.
\]
For $i\in[4]$, let $B_i$ be the $(s_i+1)$-regular
graph on $G/H$ with edges $gH\sim gdH$ for $dH\in D_i/H$.
The direction-$i$ edges of $X$ follow $B_i$ and flip the $i$th type bit,
as in~\eqref{eq:coset-edges}.
Write $X(j)$ for the $j$-faces and $S(f)\subseteq[4]$ for the
direction set of a face $f$.

For the local code system, we use the field $\E$, base codes $C_i$, and local codes $C_{\tau}$
as in Definition~\ref{def:four-dimensional-local-codes}.
Write
$
Q(f)=\{Q\in X(4):f\subseteq Q\}.
$
For every face $f$, define the set of valid local views
\begin{equation}
\label{eq:merged-spaces}
F_f=\left\{z\in \E^{Q(f)}:
z|_{Q(\tau)}\in C_{\tau}
\text{ for every }\tau\in X(3)\text{ containing }f\right\}.
\end{equation}
Thus $F_{\tau}=C_{\tau}$ on $3$-faces and $F_Q=\E$ on $4$-faces. For faces $f\subseteq g$,
restriction to $Q(g)$ maps $F_f$ into $F_g$.
The cochain spacces are
\begin{equation}
\label{eq:merged-delta}
\mathcal C^j=\bigoplus_{f\in X(j)}F_f.
\end{equation}
The coboundary maps $\delta^j:\mathcal C^j\to\mathcal C^{j+1}$, for $0\le j<4$, are defined by
\begin{align}
(\delta^j x)(g)=
\sum_{\substack{f\subseteq g\\\dim f=j}}x(f)|_{Q(g)},  \qquad \text{for } x \in \mathcal{C}^j, \ g\in X(j+1)
\end{align}

If $f\subseteq g$ and $\dim g=\dim f+2$, there are exactly two
intermediate faces. The two compositions both restrict to $Q(g)$,
so they cancel in characteristic two.
Thus $\delta^{j+1}\delta^j=0$, giving the cochain complex
\[
\mathcal C^\bullet= \mathcal C^0\xrightarrow{\delta^0}\mathcal C^1
\xrightarrow{\delta^1}\mathcal C^2
\xrightarrow{\delta^2}\mathcal C^3
\xrightarrow{\delta^3}\mathcal C^4.
\]

By Lemma~\ref{lem:local-tensor-four-dimensional}, at each vertex
$v$, under the weights $\nu_v(Q)$ and label permutations
$a_i\mapsto a_i^{2^{ht_i}}$ we can identify
\[
F_f\cong\bigotimes_{i\notin S(f)}C_i
\qquad(f\ni v).
\]
% For a fixed vertex $v$, these identifications use the same weights and labels for every face $f$ containing $v$, restriction to a superface fixes the array indices corresponding to the added directions.

Choose an $\E$-basis for each $F_f$. In these bases, we define check matrices
\[
H_Z=\delta^2,\qquad H_X=(\delta^1)^{\mathsf T}
\]
and place the physical qudits on the coordinates of $\mathcal C^2$ in the chosen bases.

Since
\[
H_XH_Z^{\mathsf T}=(\delta^2\delta^1)^{\mathsf T}=0,
\]
these matrices define a CSS code over $\E$, encoding $K$ qudits
into $N$ qudits of alphabet $\E$, where
\[
N=\dim_{\E}\mathcal C^2,\qquad
K=\dim_{\E}(\ker\delta^2/\operatorname{im}\delta^1).
\]

The following claim is immediate.
\begin{claim}
    \label{claim:check-weights}
    The row and column weights of $H_Z, H_X$ are bounded whenever $s_i,\ i \in [4]$ are bounded.
\end{claim}

\subsection{Rate}
\label{subsec:merged-rate}
We bound the quantum code rate using the argument
of~\cite[Section~4]{DLV}.

\begin{lemma}[Rate]
\label{lem:merged-rate}
Write $R_i=\dim_{\E} C_i/(s_i+1)$.
If $
R_1,R_2<\mu$ and $R_3,R_4>1-\mu$, then
$\frac{K}{N}\ge \frac{2-(1+2\mu)^4}{24}$.
\end{lemma}

\begin{proof}
It is straightforward to extend the proof of \cite[Theorem~4.1]{DLV} to our setting. They assumed the same local codes with the same block length are used throughout their cubical complex. However, their proof only relies on the local views being tensor codes, which our construction satisfies. Adapting the dimension count in \cite[Lemma~4.2 and the proof
of Corollary~3.7]{DLV} gives
$K
\ge |X(4)|\left(2-(1+2\mu)^4\right).$
On the other hand, $
N=4|X(4)|\sum_{1\le i<j\le4}R_iR_j
\le24|X(4)|.$
Combining proves the claim.
\end{proof}

\subsection{Distance and Soundness}
\label{subsec:merged-expansion}

We next apply~\cite{DLV}'s expansion analysis to the cubical complex defined in Section~\ref{subsec:cubical-construction}. They use the expansion of the base graphs of the complex along with the
product expansion of the local tensor codes and their duals to conclude the (co)systolic distance and (co)cycle expansion of the complex via a local-to-global argument, which then translates to the distance and local testability of the associated quantum code.

In Section~\ref{sec:complex}, we establish the following:

\begin{lemma}[Spectral expansion]
\label{lem:directional-ramanujan}
For our cubical complex $X$ from Lemma~\ref{lem:construction}, each $B_i$ for $i \in [4]$ is connected and $X$ is $\lambda$-expanding for $\lambda = \max_{i \in [4]}\frac{2\sqrt{s_i}}{s_i+1}$.
\end{lemma}

\begin{lemma}[Distance and soundness]
\label{thm:merged-expansion}
Fix the field $\E$ and the parameters $s_1,\ldots,s_4$ bounded independently of $N$.
For every $\rho>0$ and $\beta\ge1$, there is
$\lambda=\lambda(\rho,\beta)>0$ with the following property.
Suppose $\max_i(s_i+1)/\min_i(s_i+1)\le\beta$, every nonempty subtuple
of $(C_i)$ and of $(C_i^\perp)$ is $\rho$-product-expanding,
and $X$ has spectral expansion $\lambda$.
Then the resulting qudit CSS code has distance $D=\Omega(N)$ and
soundness $\kappa>0$.
\end{lemma}

\begin{proof} 
Again, by Lemma~\ref{lem:local-tensor-four-dimensional}, the local views
and incidence maps of our cochain complex have tensor code structures. This suffices to extend the proof of in~\cite[Corollary 3.7]{DLV}, giving linear distance and constant soundness in the coboundary direction, provided $\lambda$ is sufficiently
small in terms of $\rho, \beta$ (\cite{DLV} assume the directions have equal degrees, but unequal degrees only affect the constants only through $\beta$). To prove the boundary direction's distance and soundness,~\cite[Propostion 8.1]{DLV} performs an argument on a `dual' cochain complex where the local codes at cubes $(C_\tau)$ are replaced by their duals $(C_\tau^\perp)$. Thus, we only need to verify that the local views in this dual cochain complex are tensor codes. Indeed, $\PRS_{d_i}(\mathbb F_{s_i};\E)^\perp
=\PRS_{d'_i}(\mathbb F_{s_i};\E)$ with $d'_i =s_i-1-d_i$, and the dual codes  $(\PRS_{d'_i}(\mathbb F_{s_i};\E))$ satisfy the requirements in Definition~\ref{def:four-dimensional-local-codes}. Thus Lemma~\ref{lem:local-tensor-four-dimensional} applies.
\end{proof}

\subsection{Product Expansion of Reed-Solomon Codes}
In this section, we state our conjecture on the product expansion of Reed-Solomon codes.
For a finite set $S\subseteq \E$ and $0\le d<|S|$, write
\[
\RS_d(S;\E)=\{(f(x))_{x\in S}:f\in \E[T],\ \deg f\le d\}.
\]
Thus the subscript denotes the maximum polynomial degree, as in the PRS notation.

Fix $k\ge2$, finite sets $S_1,\ldots,S_k$, and write
\[
 \Omega=\prod_{i=1}^kS_i,\qquad
 N=|\Omega|=\prod_{i=1}^k|S_i|.
\]
An $i$-axis line is obtained by fixing every coordinate except the
$i$th.  For a code $C_i\subseteq\F_q^{S_i}$, let
\[
 L_i(C_i)
 =
 \F_q^{S_1}\otimes\cdots\otimes
 \F_q^{S_{i-1}}\otimes C_i\otimes
 \F_q^{S_{i+1}}\otimes\cdots\otimes\F_q^{S_k}.
\]
For $M_i\in L_i(C_i)$, let $\ell_i(M_i)$ be the number of nonzero $i$-axis lines.  

\begin{definition}[Product expansion]
\label{def:product-expansion}
A tuple $(C_i)_{i\in I}$, where $I\subseteq[k]$, is
$\rho$-product-expanding if every
$M\in\sum_{i\in I}L_i(C_i)$ has a decomposition
$M=\sum_{i\in I}M_i$, with $M_i\in L_i(C_i)$, such that
\[
 \wt(M)\ge\rho\sum_{i\in I}|S_i|\ell_i(M_i).
\]
\end{definition}

We now state our conjecture for product expansion of Reed-Solomon codes over binary extension fields,
\begin{conj}[Product expansion over Binary Extension Fields]
\label{conj:binary-pe}
For every $r\ge 2$, $\eta\in(0,1)$, and $\beta\ge 1$, there exists
$\rho=\rho(r,\eta,\beta)>0$ with the following property. Let $\E=\F_{2^e}$, and let $m_1,\ldots,m_r$ be positive integers such
that $2m_i\mid e$ for every $i$. Set $s_i=2^{m_i}$ and define\footnote{Since $2m_i\mid e$, we have $s_i+1\mid s_i^2-1\mid 2^e-1$.
Thus $|S_i|=s_i+1$ because $\E^\times$ is cyclic.}
\[
S_i \coloneqq \{x\in \E^\times : x^{s_i+1}=1\},
\qquad
|S_i|=s_i+1.
\]
Suppose that $s_1+1,\ldots,s_r+1$ are pairwise coprime and satisfy $\max_{i\in[r]}(s_i+1) \le \beta\min_{i\in[r]}(s_i+1).$ Then, for every choice of integers
$0\le d_i$ with $d_i+1\le (1-\eta)(s_i+1),$
the tuple
$\bigl(\RS_{d_i}(S_i;\E)\bigr)_{i=1}^r$
is $\rho$-product-expanding.
\end{conj}

\begin{lemma}[Product expansion of projective Reed--Solomon codes]
\label{lem:prs-rs-pe}
Let $\E,s_i,S_i$ be as in Conjecture~\ref{conj:binary-pe}, and let
$C_i=\PRS_{d_i}(\F_{s_i};\E)$, where $0\le d_i\le s_i-1$. Assuming
Conjecture~\ref{conj:binary-pe}, if
$
\eta\le\frac{d_i+1}{s_i+1}\le1-\eta$ for every $i$,
then all nonempty subtuples of $(C_i)$ and of $(C_i^\perp)$ are
$\rho$-product-expanding for a common $\rho>0$ depending only on
$r,\eta,\beta$.
\end{lemma}

\begin{proof}
We will show that $C_i$ and $C_i^\perp$ are equivalent, by coordinate permutations
and nonzero coordinate scalings, to
$\RS_{d_i}(S_i;\E)$
and $\RS_{s_i-1-d_i}(S_i;\E)$
respectively. Provided that $\eta\le\frac{d_i+1}{s_i+1}\le1-\eta$, the latter Reed-Solomon code tuples have product expansion from the conjecture. The equivalence then transfers this product expansion to that of the projective RS code tuples because product expansion only concerns Hamming weights.

We now show the equivalence between $C_i$ and $\RS_{d_i}(S_i;\E)$.
Fix $i$ and write $s=s_i$, $d=d_i$, and $S=S_i$. First, we identify the evaluation sets of PRS and RS codes. Choose
$\theta\in\F_{s^2}\setminus\F_s$ and define
\[
\beta(x)=\frac{x-\theta}{x-\theta^s},
\qquad \beta(\infty)=1.
\]
For $x\in\F_s$, we have
$\beta(x)^s=\frac{x-\theta^s}{x-\theta}$ (this follows from simple algebra, using $s$ is a power of characteristic and $x^s=x$, $\theta^{s^2}=\theta$), so $\beta(x)^{s+1}=1 \implies \beta(x) \in S$.
Since $\theta\ne\theta^s$, this is an invertible M\"obius
transformation. It therefore maps $\PP^1(\F_s)$ bijectively
onto $S$, both sets having size $s+1$.

Next, we change polynomial variables using the linear transformation
\[
B=\begin{pmatrix}1&-\theta\\1&-\theta^s\end{pmatrix}.
\]
For a homogeneous degree-$d$ polynomial $f$, set
$F=f\circ B^{-1}$ and $g(T)=F(T,1)$. The matrix $B$ is invertible, thus this gives a bijection between homogeneous degree-$d$
polynomials $f$ over $\E$ and univariate polynomials $g$ of degree
at most $d$.

Finally, let $\lambda_x=x-\theta^s$ for finite $x$ and
$\lambda_\infty=1$. These scalars are nonzero, and
$Br_x=\lambda_x r_{\beta(x)}$. Hence
\[
f(r_x)=F(Br_x)
=\lambda_x^d F(r_{\beta(x)})
=\lambda_x^d g(\beta(x)).
\]
Thus dividing coordinate $x$ by $\lambda_x^d$ and relabeling
it by $\beta(x)$ transforms the PRS codeword into
$(g(z))_{z\in S}$. Since every polynomial $g$ of degree at
most $d$ occurs, the resulting code is exactly
$\RS_d(S;\E)$.

The equivalence of $C_i^\perp$ and $\RS_{s_i-1-d_i}(S_i;\E)$ follows exactly the same argument by using Claim~\ref{lem:prs-dual}.
\end{proof}

\subsection{Proof of Main Theorem}
We prove the main theorem in this section by combining the results in preceding sections and choosing parameters.

We first choose the four base local codes satisfying the conditions required in the preceding lemmas.

\begin{definition}[Parameters for local codes]
\label{def:main-parameters}
Let $L$ be a sufficiently large positive multiple of $128$, and set
\[
(m_1,m_2,m_3,m_4)=(L+8,L+16,L+32,L+64),
\qquad s_i=2^{m_i}.
\]
Over the common field
$\E=\F_{2^{\,2\operatorname{lcm}(m_1,m_2,m_3,m_4)}}$, take $ C_i=\PRS_{d_i}(\F_{s_i};\E)$ where
\begin{equation}\label{eq:local-degree-condition}
d_i=
\begin{cases}
(s_i-1)/17,&i=1,2,\\
16(s_i-1)/17,&i=3,4.
\end{cases}
\end{equation}
These base codes satisfy Definition~\ref{def:four-dimensional-local-codes}
with $h=4$.
\end{definition}

We explain the choices. Since $8\mid m_i$, we have $17\mid s_i-1$, and
thus the degrees are integral and satisfy $s_i-1\mid17d_i$ as
required for the divisibility condition. We pick $h=4$ the smallest choice,
so that the first two code rates are below $1/16$ and the last two are above
$15/16$, giving the separation needed for the rate bound of Lemma~\ref{lem:merged-rate} to be nontrivial.

All four rates lie in $[1/17,16/17]$, so both the codes and their duals have rates bounded away from one. The largest powers of two dividing $m_1,m_2,m_3,m_4$ are
$8,16,32,64$, respectively. Moreover, 
$\frac{\max_i(s_i+1)}{\min_i(s_i+1)}<2^{56}.$ For Conjecture~\ref{conj:binary-pe} and Lemma~\ref{lem:prs-rs-pe} to apply, we finally need $s_1+1,\ldots,s_r+1$ to be pairwise coprime. Indeed, consider a pair $1\leq i<j \leq 4$ and the equality $s_i^{m_j}=s_j^{m_i}$. Taking $2^{i+2}$-root both sides give $ s_i^{\,L/2^{i+2}+2^{j-i}} = s_j^{\,L/2^{i+2}+1}.$ Because $128 \mid L$, the left exponent is even while the right exponent is odd. Now suppose for contradiction there existed an odd prime $p$ such that $p$ divides both $s_i + 1$ and $s_j+1$, then $s_i=-1 \mod p$ and $s_j=-1 \mod p$. Then $s_i^{\,L/2^{i+2}+2^{j-i}} = s_j^{\,L/2^{i+2}+1}$ implied $1=-1 \mod p$, a contradiction.

Note that increasing $L$ makes the spectral bounds Lemma~\ref{lem:directional-ramanujan} tend to zero while keeping the rate bounds and the degree ratio fixed. After choosing $L$, we keep it and $\E$ fixed as the global cubical complex grows.

With the above parameter choices, we can establish the main theorem.

\begin{theorem}[Main Theorem]\label{cor:main}
Assuming Conjecture~\ref{conj:binary-pe}, there exists
an infinite family of quantum CSS codes on $N$ qubits with $K=\Omega(N)$,
distance $D=\Omega(N)$, constant soundness, and constant check weights.
\end{theorem}

\begin{proof}
A qudit CSS code over alphabet $\E$ with the stated properties directly follows from our construction. In particular, Lemma~\ref{lem:prs-rs-pe}, with $\eta=1/17$ and $\beta=2^{56}$,
gives a product-expansion constant $\rho>0$ independent of $L$
for all nonempty subtuples of the local codes and their duals.
By Lemma~\ref{lem:directional-ramanujan}, choosing $L$ sufficiently
large makes every $2\sqrt{s_i}/(s_i+1)$ smaller than the threshold
in Lemma~\ref{thm:merged-expansion}.
Fix such an $L$ and use the arbitrarily large complexes of
Lemma~\ref{lem:construction} together with the local codes above to construct the cochain complex as described in Section~\ref{subsec:cubical-construction}. Then
Lemma~\ref{lem:local-tensor-four-dimensional} and
Claim~\ref{claim:check-weights} establishes CSS construction with
bounded check weights and degrees.
Lemma~\ref{lem:merged-rate}, with $\mu=1/16$, gives $K=\Omega(N)$,
and Lemma~\ref{thm:merged-expansion} gives $D=\Omega(N)$ and
constant soundness.

To obtain a qubit code, we expand the cochain maps $\delta^2$, $\delta^1$ in an $\F_2$-basis of $\E$ and and define qubit check matrices using the expanded maps. This gives $N \log_2|\E|$ physical and $K \log_2|\E| $ logical qubits. The distance is at least $D$ and soundness decreases by at most a factor $O(\log_2|\E|)$. Each row and column weight of the check matrices increases by at most a factor $\log_2|\E|$. Since $\E$ is fixed, all claimed parameters follow.
\end{proof}

\section{The Non-Abelian Cubical Complex}\label{sec:complex}

The construction of the complexes below are due to Rungtanapirom, Stix, and Vdovina
\cite[\S2.1-2.4; \S6.2]{RSV19}. We begin by describing the infinite trees underlying the cubical complex $X$ and how to assign invertible $2 \times 2$ matrices to the edges incident to each vertex, which enables us to exhibit the matrices arising in property (4) of Lemma~\ref{lem:construction}. 

First, we set up some notation. We choose positive integers $m_1,\ldots,m_4$ divisible by $8$ and let
$s_i=2^{m_i}$. The field $K_i=\F_{s_i}((\pi_i))$ consists of formal
series $\sum_{n\ge N}a_n\pi_i^n$, where $N\in\mathbb Z$ and
$a_n\in\F_{s_i}$. The ring $\cO_i=\F_{s_i}[[\pi_i]]$ consists of series with only terms with
nonnegative exponents. For a nonzero series $a$, its valuation $\val_i(a)$ is the smallest
exponent with nonzero coefficient and reduction modulo $\pi_i$ keeps
only the constant coefficient of a series in $\cO_i$; we denote
this operation by $\overline{a} :=a\pmod{\pi_i}$. 

\paragraph{Bruhat-Tits tree and corresponding projective line edge labels.}
An \emph{$\cO_i$-lattice} is defined by
$L=\cO_i p\oplus\cO_i q$, where $p,q\in K_i^2$ are linearly
independent over $K_i$. We call $(p,q)$ an \emph{ordered basis} of $L$
and write $L=U\cO_i^2$ for the matrix $U=(p\;q)$.
We identify lattices that differ by multiplication by a nonzero
scalar in $K_i$, and write $[L]$ for the corresponding equivalence class.

The vertices of $\cT_i$ are these classes. Two vertices are adjacent
if they admit representatives $L,L'$ satisfying
\begin{equation}\label{eq:adjacent-lattices}
\pi_iL\subsetneq L'\subsetneq L.
\end{equation}
This graph is the \emph{Bruhat-Tits tree}
\cite[Chapter~II]{S80}. We define $\cT=\prod_{i=1}^4\cT_i$, where here the $\prod$ denotes the Cartesian product of graphs.

For a fixed $L$, every neighboring class
has a unique representative $L'$ satisfying
\eqref{eq:adjacent-lattices}. The map $L'\mapsto L'/\pi_iL$
identifies these neighbors with the one-dimensional subspaces
of $L/\pi_iL$. An ordered basis of $L$  in turn identifies $L/\pi_iL$ with $\F_{s_i}^2$
by reducing the two coordinates modulo $\pi_i$, so we can think of the neighbors of $L$ as corresponding to one-dimensional subspaces of this copy of $\F_{s_i}^2$.
The set of one-dimensional subspaces of $\F_{s_i}^2$ is called
the \emph{projective line}, denoted $\PP^1(\F_{s_i})$.
It has $s_i+1$ elements, and we can think of each of these elements as indexing a neighbor of $[L]$. Henceforth when we refer to an edge label we mean the corresponding projective line. 

A vertex $v\in\cT$ is a tuple $([L_1], [L_2], [L_3],[L_4])$. We describe it in coordinates by choosing a representative lattice and an ordered basis in each factor.  A direction-$i$ edge changes only the $i$th coordinate to a neighboring lattice class. For this edge, we can consider its label as a change of basis matrices that correspond to a choice of basis $V=(p\;q)$ and $V\diag(1,\pi_i)$ for its respective endpoints. These basis may differ from the bases chosen independently at the vertices; we will soon introduce \emph{comparison matrices} that describe how to relate these two choices.

\paragraph{Vertex types and group actions.}
For $L=U\cO_i^2$ for some $U\in\GL_2(K_i)$, we define the \emph{type} of $[L]$ to be
$\val_i(\det U)\bmod2$. A change of basis for the lattice multiplies
$\det U$ by an element of valuation zero, while rescaling $L$
changes its type by an even number. That is,
the type is well-defined and independent of the choice of representative of $L$.
For adjacent vertices, we can associate to the lattices at each of them the bases $(p,q)$
and $(p,\pi_iq)$. Their determinants differ by a factor of
$\pi_i$, so their determinant valuations differ by one.
In particular, it follows that every edge joins vertices of opposite types.

We write $\GL_2(K_i)$ for the invertible $2\times2$ matrices over
$K_i$. The group $\PGL_2(K_i)$ identifies matrices that differ
by a nonzero scalar and acts on tree vertices by
$[L]\mapsto[AL]$.  

% Its type-preserving subgroup is
% $\cG_i^+=\{[A]:\val_i(\det A)\equiv0\pmod2\}$.

% We define $\cT=\prod_{i=1}^4\cT_i$ and let $\cG=\prod_i\cG_i^+$.
% A vertex of $\cT$ has the four types of its component vertices.
% A cube is a product in which each factor is either a vertex
% or an edge.

\paragraph{Choosing coordinate bases for the lattices at endpoints of an edge.}
We fix direction $i$ and consider matrices $U,V\in \GL_2(K_i)$ whose corresponding lattices
represent the same tree vertex. That is, there exists $a\in K_i^\times$
such that $aV\cO_i^2=U\cO_i^2$. Then the matrix $M=U^{-1}aV$
and its inverse have entries in $\cO_i$, since they compare
two bases of the same lattice. We denote this group of matrices
by $\GL_2(\cO_i)$.

The reduction $\overline M\in\GL_2(\F_{s_i})$ expresses the edge labels obtained from $V$ in the coordinates obtained from $U$.
A convenient choice of normalization is to record this action using a determinant-one matrix: every nonzero element of $\F_{s_i}$ has a unique square root, so $\nu_i(R)=(\det R)^{-1/2}R$ is the unique determinant-one scalar multiple of $R$ which acts on one-dimensional subspaces in the same way as $R$. To that end, we define the \emph{comparison matrices}
\begin{equation}\label{eq:coordinate-comparison}
\mathsf C_i(U,V)
=\nu_i\!\left(\overline{U^{-1}aV}\right)
\in\SL_2(\F_{s_i}).
\end{equation}

To check that $\mathsf C_i(U,V)$ is well-defined, we note that any two admissible choices of $a$ differ by an invertible element
of $\cO_i$. Their reduction $\pmod \pi_i$ differ by a nonzero scalar,
which is then removed after the appropriate normalization. It follows that
$\mathsf C_i(U,V)$ is independent of $a$. Furthermore, rescaling either basis
matrix by an element of $K_i^{\times}$ also leaves this comparison unchanged.

% \paragraph{Computing the comparison matrix.}
% We compute $U^{-1}V$ and let $m$ be the minimum valuation of its
% nonzero entries. We then multiply by $\pi_i^{-m}$, reduce each
% entry modulo $\pi_i$, and normalize the determinant to one.
% To justify this procedure, we first write
% $U^{-1}V=bS$ for some $b\in K_i^\times$ and
% $S\in\GL_2(\cO_i)$. Some entry of $S$ has nonzero constant
% coefficient, so $m=\val_i(b)$ and
% $\pi_i^{-m}U^{-1}V\in\GL_2(\cO_i)$, as required.

\paragraph{Composition of change of basis and lifting from $X$ into $\cT$.}
For basis matrices $U,V,W$ representing the same tree vertex,
and $A\in\GL_2(K_i)$, by thinking of the $\mathsf C_i(U,V)$ as change of basis, we have the following composition laws of these comparison matrices
\begin{equation}\label{eq:coordinate-composition}
\begin{aligned}
\mathsf C_i(U,V)\mathsf C_i(V,W)&=\mathsf C_i(U,W),\\
\mathsf C_i(AU,AV)&=\mathsf C_i(U,V).
\end{aligned}
\end{equation}

For $X=\Gamma\backslash\cT$, a \emph{lift} of a vertex, edge,
or cube is a corresponding vertex, edge, or cube in $\cT$
mapping to it. The construction ensures that $\Gamma\le\cG$
acts freely on vertices: only the identity fixes a vertex.
Together with type preservation, this makes the element of
$\Gamma$ relating two lifts of a face unique.
We can transport the chosen lattice bases to the entire under this action.
Since changing a common lift then multiplies both basis matrices
by the same matrix in each factor, it follows from the second identity in
\eqref{eq:coordinate-composition} that we may compute the change of basis for faces of $X$ on any suitable lift in $\cT$.

\subsection{Proof of Lemma~\ref{lem:construction}}

Before proving Lemma 4.1, we give an informal sketch of its four properties. Properties (1) and (2) are essentially standard features of the construction. For (3) and (4), we lift to the product of trees, where the relevant matrices describe changes of lattice basis. 

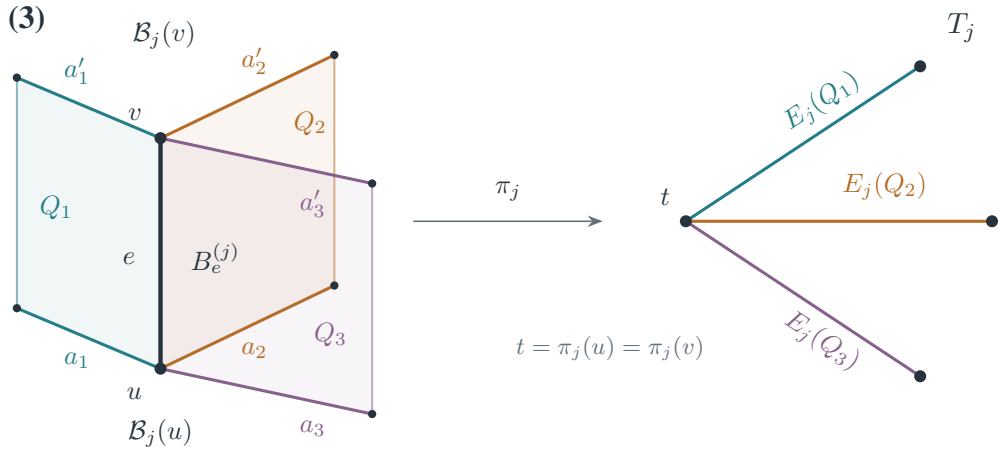
\[\begin{tikzpicture}
  \path[use as bounding box] (0,-1.10) rectangle (13.9,7.45);
  \node[title] at (.15,7.12) {(3)};

  \coordinate (u) at (2.30,2.50);
  \coordinate (v) at (2.30,5.55);
  \coordinate (u1) at (.40,3.30);
  \coordinate (v1) at (.40,6.35);
  \coordinate (u2) at (4.60,3.60);
  \coordinate (v2) at (4.60,6.65);
  \coordinate (u3) at (5.10,1.90);
  \coordinate (v3) at (5.10,4.95);

  % Light fills separate the pages while keeping the common spine visible.
  \fill[teal,opacity=.065] (u)--(u1)--(v1)--(v)--cycle;
  \fill[ochre,opacity=.065] (u)--(u2)--(v2)--(v)--cycle;
  \fill[plum,opacity=.065] (u)--(u3)--(v3)--(v)--cycle;
  \draw[edge,teal!55] (u1)--(v1);
  \draw[edge,ochre!55] (u2)--(v2);
  \draw[edge,plum!55] (u3)--(v3);

  % A matching pair of transverse edges in each page.
  \draw[accent,teal] (u)--node[pos=.58,below=3pt] {$a_1$}(u1);
  \draw[accent,teal] (v)--node[pos=.58,above=3pt] {$a'_1$}(v1);
  \draw[accent,ochre] (u)--node[pos=.54,below=3pt] {$a_2$}(u2);
  \draw[accent,ochre] (v)--node[pos=.54,above=3pt] {$a'_2$}(v2);
  \draw[accent,plum] (u)--node[pos=.72,below=3pt] {$a_3$}(u3);
  \draw[accent,plum] (v)--node[pos=.72,below=3pt] {$a'_3$}(v3);

  \draw[line width=1.65pt] (u)--node[pos=.48,left=5pt] {$e$} node[pos=.48,right=7pt] {$B_e^{(j)}$}(v);
  \foreach \p in {u1,v1,u2,v2,u3,v3}
    \node[point,inner sep=1.1pt] at (\p) {};
  \node[point,label={[label distance=2pt]below left:$u$}] at (u) {};
  \node[point,label={[label distance=2pt]above left:$v$}] at (v) {};
  \node[text=teal] at (.92,4.63) {$Q_1$};
  \node[text=ochre] at (4.28,5.75) {$Q_2$};
  \node[text=plum] at (4.55,2.95) {$Q_3$};

  \node at (2.35,6.91) {$\mathcal B_j(v)$};
  \node at (2.30,1.65) {$\mathcal B_j(u)$};

  \draw[guide] (5.65,4.45)--node[above=4pt] {$\pi_j$}(8.15,4.45);

  % The spine collapses in the transverse tree factor; each page gives a
  % (possibly different) edge of the star at t.
  \coordinate (t) at (9.25,4.45);
  \coordinate (t1) at (12.35,6.50);
  \coordinate (t2) at (13.30,4.45);
  \coordinate (t3) at (12.35,2.40);
  \draw[accent,teal] (t)--node[pos=.64,above,sloped] {$E_j(Q_1)$}(t1);
  \draw[accent,ochre] (t)--node[pos=.65,above=4pt] {$E_j(Q_2)$}(t2);
  \draw[accent,plum] (t)--node[pos=.64,below,sloped] {$E_j(Q_3)$}(t3);
  \foreach \p in {t,t1,t2,t3} \node[point] at (\p) {};
  \node[above left=3pt] at (t) {$t$};
  \node[note] at (8.25,2.77)
    {$t=\pi_j(u)=\pi_j(v)$};
  \node[font=\normalsize] at (12.90,7.01) {$T_j$};

  \node[figcaption,text width=13.6cm] at (.15,.85) {
    In the lift on the RHS, each colored square $Q_r$ contains the fixed edge
    $e=[u,v]$. Its two transverse $j$-edges project to the same tree edge
    $E_j(Q_r)$; $a_r$ and $a'_r$ are their labels in the bases at $u$ and $v$.
    Since these bases are fixed, one matrix $B_e^{(j)}$ relates the labels
    across all the squares.
  };
\end{tikzpicture}\]

For (3), we fix an edge e and consider the squares containing it in another direction. The opposite edges of each square project to the same tree edge, whose labels we read using the bases at the two endpoints of $e$. Consequently, the same change-of-basis matrix works for the entire book of squares around $e$ because those vertex bases are fixed independently of the square. This is essentially the content of the $B_e^{(j)}$ matrices. 

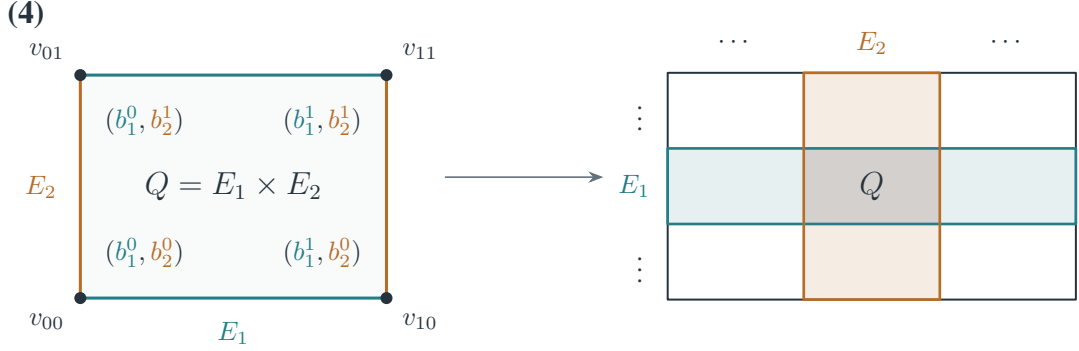
\[ \begin{tikzpicture}
  \path[use as bounding box] (0,1.00) rectangle (14.9,8.05);
  \node[title] at (.15,7.85) {(4)};

  \coordinate (v00) at (1.25,4.10);
  \coordinate (v10) at (5.30,4.10);
  \coordinate (v01) at (1.25,7.05);
  \coordinate (v11) at (5.30,7.05);
  \fill[ink,opacity=.025] (v00)--(v10)--(v11)--(v01)--cycle;
  \draw[accent,teal] (v00)--(v10) (v01)--(v11);
  \draw[accent,ochre] (v00)--(v01) (v10)--(v11);
  \foreach \p in {v00,v10,v01,v11} \node[point] at (\p) {};
  \node[below left=3pt] at (v00) {$v_{00}$};
  \node[below right=3pt] at (v10) {$v_{10}$};
  \node[above left=3pt] at (v01) {$v_{01}$};
  \node[above right=3pt] at (v11) {$v_{11}$};
  \node[text=teal,below=5pt] at ($(v00)!.5!(v10)$) {$E_1$};
  \node[text=ochre,left=5pt] at ($(v00)!.5!(v01)$) {$E_2$};
  \node[font=\large] at (3.275,5.575) {$Q=E_1\times E_2$};

  % The coordinate-i endpoint basis is reused along the other direction.
  \node[anchor=south west] at (1.44,4.38)
    {$(\textcolor{teal}{b_1^0},\textcolor{ochre}{b_2^0})$};
  \node[anchor=south east] at (5.11,4.38)
    {$(\textcolor{teal}{b_1^1},\textcolor{ochre}{b_2^0})$};
  \node[anchor=north west] at (1.44,6.77)
    {$(\textcolor{teal}{b_1^0},\textcolor{ochre}{b_2^1})$};
  \node[anchor=north east] at (5.11,6.77)
    {$(\textcolor{teal}{b_1^1},\textcolor{ochre}{b_2^1})$};

  \draw[guide] (6.08,5.70)--(8.17,5.70);

  % Columns are indexed by E_2; rows are indexed by E_1.
  % The chosen row and column meet at the entry Q=E_1 x E_2.
  \fill[teal,opacity=.11] (9.02,5.08) rectangle (14.42,6.08);
  \fill[ochre,opacity=.13] (10.82,4.08) rectangle (12.62,7.08);
  \fill[plum,opacity=.12] (10.82,5.08) rectangle (12.62,6.08);
  \draw[edge] (9.02,4.08) rectangle (14.42,7.08);
  \foreach \x in {10.82,12.62}
    \draw[draw=hairline,line width=.5pt] (\x,4.08)--(\x,7.08);
  \foreach \y in {5.08,6.08}
    \draw[draw=hairline,line width=.5pt] (9.02,\y)--(14.42,\y);
  \draw[teal,line width=1pt] (9.02,5.08) rectangle (14.42,6.08);
  \draw[ochre,line width=1pt] (10.82,4.08) rectangle (12.62,7.08);
  \node[font=\large] at (11.72,5.58) {$Q$};
  \node at (9.92,7.48) {$\cdots$};
  \node[text=ochre] at (11.72,7.48) {$E_2$};
  \node at (13.52,7.48) {$\cdots$};
  \node at (8.65,6.58) {$\vdots$};
  \node[text=teal] at (8.58,5.58) {$E_1$};
  \node at (8.65,4.58) {$\vdots$};

  \node[figcaption,text width=14.6cm] at (.15,2.95) {
    The corner pairs $(b_1^a,b_2^b)$ are induced by ``change of basis'' matrices on
    $E_1$ and $E_2$. The square $Q=E_1\times E_2$ corresponds to the
    common entry of the teal row ($E_1$ fixed) and orange column
    ($E_2$ fixed). Both checks use the same basis at this entry.
  };

\end{tikzpicture}\]

For (4), we lift a square into $\cT$ as $\widetilde Q=E_1\times E_2$. For each $E_i$, we choose corresponding matrices that allows us to read off lattices corresponding to its endpoints. In the coordinates given by $V_i$, these lattices have the standard forms $\cO_i^2$ and $\diag(1,\pi_i)\cO_i^2$. At each vertex $v$ of the $Q$, the matrix $F_i(v,Q)$ essentially records the comparison between the (relative) choice of basis for $v$ as read off from the matrices assigned for $E_i$ and the choice of basis for the lattice representing $v_i$. The composition identities give the compatibility relations for both incident edges to $v$ using these same matrices (which is why we can construct the same weight $\nu_v(Q)$  for both columns and rows in Section~\ref{sec:local-view}).

\begin{proof}[Proof of Lemma~\ref{lem:construction}]
\emph{(1) Quotient structure of $X$.}
This follows essentially from the machinery already developed in \cite{RSV19}. The construction in \cite[\S~2.1-2.4]{RSV19}
supplies a group $\Lambda$ acting transitively on the
vertices of the product of trees $\cT$. The stabilizer $H_0$ of a base vertex (which can be thought of as the origin) is isomorphic to $S_3$, by
\cite[Proposition~2.1 and the proof of Proposition~2.21(1)]{RSV19}. 

% The strong approximation argument in \cite[Remark~6.10]{RSV19}
% gives a surjective reduction map from the type-preserving subgroup
% of $\Lambda$ onto $G=\PGL_2(\FF_{2^r})$, using the equality
% $\operatorname{PSL}_2(\FF_{2^r})=\PGL_2(\FF_{2^r})$.
% We let $\Gamma$ be its kernel and obtain a finite cubical complex
% $X=\Gamma\backslash\cT$. 

It is possible to embed $H_0$ as a subgroup conjugate to
$H=\PGL_2(\FF_2)$, and a standard argument allow us to identify $X$ with $(G/H)\times\{0,1\}^4$. In slightly more detail, let $\Lambda$ act on $\cT$ as in \cite{RSV19}, where $\Lambda$ is defined using an appropriate quarternion algebra, and let $\theta:\Lambda\to G$ be the group homomorphism corresponding to reduction at a irreducible polynomial of $\F_2[t]$ of sufficiently large degree \cite[Remark 6.10]{RSV19}. We also define $\tau:\Lambda\to\{0,1\}^4$ which records how the action by $\Lambda$ changes the vertex types of the tree,
and let $\Gamma=\ker\theta\cap\ker\tau$.
The ideas in \cite[Remark~6.10]{RSV19} show that
$\theta(\ker\tau)=G$. Since $\Lambda$ is vertex transitive, this shows in turn  that $(\theta,\tau)$ is onto. That is, each coset $G/H$ of every $\type$ vector occurs.
Moreover, the aforementioned reduction map can be shown to embed $H_0$ as a subgroup conjugate to
$H=\PGL_2(\FF_2)$, and $\tau(H_0)=0$. Consequently,
\[
 V(X)=\Gamma\backslash\Lambda/H_0
 \cong (G\times\{0,1\}^4)/(H\times\{0\})
 \cong (G/H)\times\{0,1\}^4.
\]
We may choose the irreducible polynomial of sufficiently large degree to not be amongst the finitely many
nonidentity elements so that $\theta(\lambda)\ne1$ for every nonidentity $\lambda\in\Lambda$ sending the base vertex $o$ to $\lambda o$ of distance at most 8 from it. By using the definition of $\Gamma$ as the intersection of kernels and the vertex transitivity of $\Lambda$, we can show that no two vertices at distance at most 8 are merged in the quotient, which ensures that the cubical structure around every vertex in the quotient $X$ is the same as its corresponding vertex in the infinite graph $\cT$.

The relations $D_iD_j=D_jD_i$ for $i\ne j$
come from the product structure of $\cT$: two successive steps in
distinct directions can be taken in either order to reach the same
vertex.

\emph{(2) Projective line labels for edges.}
In each direction the local construction is the same Bruhat-Tits
description that underlies the LPS graphs \cite{lubotzky1988ramanujan}: links are the lines in its two-dimensional residue space, showing the identification of the incident edges with $\PP^1(\F_{s_i})$. (Alternatively, as we have already observed this also follows from \cite[Chapter II]{S80}.)

\emph{Interlude: Lifting from the complex into the tree.}
For each vertex $v$ of $X$, we choose one lift $\widetilde v$ and,
in each factor $i$, a basis matrix $U_{\widetilde v,i}$ for a lattice
representing its $i$th coordinate. We extend these choices to the
other lifts by the action of $\Gamma$. More precisely, if
$A_i\in\GL_2(K_i)$ represents the $i$th component of
$\gamma\in\Gamma$, we set
$U_{\gamma\widetilde v,i}=A_iU_{\widetilde v,i}$. As observed in the preamble, we may carry out the computations on lifted vertices, edges, and cubes in $\cT$. %With the bases transported as above, the second identity in \eqref{eq:coordinate-composition} shows that changing the lift leaves the basis change comparison matrices unchanged, so they are well defined after descending to $X$.

\emph{(3) The matrix attached to a directed edge.}
Fix a direction-$i$ edge $e:v\to w$ and lift it to
$\widetilde e:\widetilde v\to\widetilde w$. For $j\ne i$ the endpoints
have the same $j$th tree coordinate. Hence their factor-$j$ basis matrices
span homothetic lattices, and we define
\begin{equation}\label{eq:edge-change}
 B_e^{(j)}=\mathsf C_j(U_{\widetilde w,j},U_{\widetilde v,j}).
\end{equation}
This matrix converts the direction-$j$ labels at $v$ into the coordinates
used at $w$. It depends only on the directed edge $e$ and the fixed vertex
coordinates.

In a lifted square containing $e$, the two opposite direction-$j$
edges project to the same edge of $\cT_j$. Their labels $a$ at $v$
and $\beta_e^{(j)}(a)$ at $w$ describe this edge in the respective
vertex coordinates. Since $B_e^{(j)}$ converts the coordinates at
$v$ to those at $w$, its definition in \eqref{eq:edge-change} gives
$\beta_e^{(j)}(a)=\langle B_e^{(j)}r_a^{(j)}\rangle$.
For the reverse edge $\bar e$, we exchange the two bases, giving
$B_{\bar e}^{(j)}=(B_e^{(j)})^{-1}$.

\emph{(4) Change of coordinates for cubes.}

We write $T_i=\diag(1,\pi_i)$. For a tree edge oriented from type zero
to type one, we choose representatives
$\pi_iL^0\subsetneq L^1\subsetneq L^0$.
A basis $V=(p\;q)$ of $L^0$ is \emph{adapted} to the edge if
$L^1=\cO_i p\oplus\pi_i\cO_i q$; its endpoint bases are $V$ and $VT_i$.
Such a basis exists by lifting a basis of $L^0/\pi_iL^0$ whose first
vector spans $L^1/\pi_iL^0$.

For each cube $Q$, we choose one lift $\widetilde Q$ and adapted bases
$V_{\widetilde Q,i}$ for its projected tree edges. For $v\in Q$, we define, in a common lift,
\begin{equation}\label{eq:frames}
 F_i(v,Q)=\mathsf C_i\bigl(
 U_{\widetilde v,i},V_{\widetilde Q,i}T_i^{\type_i(v)}\bigr).
\end{equation}
The invariance in \eqref{eq:coordinate-composition} makes this independent
of the choice of lift. In the adapted endpoint bases, the projective lines corresponding to the edges are given by
$\langle(1,0)^{\mathsf T}\rangle$ and $\langle(0,1)^{\mathsf T}\rangle$
respectively, since the type 0 vertex as viewed from type 1 is represented by $\pi_iL^0$.
It follows that column $\type_i(v)+1$ of $F_i(v,Q)$ spans the edge label.

For a direction-$i$ edge $e:v\to w$ in $Q$ and $j\ne i$, the cube
uses the same factor-$j$ endpoint basis at $v$ and $w$.
The composition identity gives
\begin{equation}\label{eq:transverse}
 F_j(w,Q)
 =\mathsf C_j(U_{\widetilde w,j},U_{\widetilde v,j})F_j(v,Q)
 =B_e^{(j)}F_j(v,Q).
\end{equation}

We now orient $e:v\to w$ from type zero to type one in direction $i$, and 
choose a corresponding matrix $V_{\widetilde e,i} = (p,q)$ (which we think of as assigning its two endpoints the choice of basis $L^0 = \cO_ip \oplus \cO_i q$ and $L^1 = \cO_i p \oplus \cO_i(\pi_i q)$)  independently of $Q$, and define
\begin{equation}\label{eq:edge-reference}
 P_{v,e}=\mathsf C_i(U_{\widetilde v,i},V_{\widetilde e,i}) \quad \text{and}\quad
 P_{w,e}=\mathsf C_i(U_{\widetilde w,i},V_{\widetilde e,i}T_i).
\end{equation}
For $Q\supset e$, we lift both into the same edge in $\cT$ and rescale
$V_{\widetilde Q,i}$ (since the lattices are defined up to homothety) so that both type-zero lattices equal
$L^0=V_{\widetilde e,i}\cO_i^2$.
Their type-one lattices also agree, by uniqueness of the neighbor
representative between $\pi_iL^0$ and $L^0$.
The matrix $M=V_{\widetilde e,i}^{-1}V_{\widetilde Q,i}$ lies in
$\GL_2(\cO_i)$ and preserves $T_i\cO_i^2$. As $(1 \hspace{0.8em} 0)^\top \in T_i \cO_i^2$, it follows that $M(1 \hspace{0.8em} 0)^T \in T_i \cO_i^2$ so its lower-left entry
is divisible by $\pi_i$. We obtain
\begin{equation}\label{eq:M-rep}
 M=\begin{pmatrix}a&b\\\pi_i c&d\end{pmatrix},
 \qquad
 T_i^{-1}MT_i=\begin{pmatrix}a&\pi_i b\\c&d\end{pmatrix},
 \qquad a,b,c,d\in\cO_i.
\end{equation}
Both reductions $\pmod{\pi_i}$ have determinant $\bar a\bar d\ne0$.
Normalizing the determinant to one, we get
\begin{equation}\label{eq:edge-bases}
 \begin{aligned}
 P_{v,e}^{-1}F_i(v,Q)
 &\underset{\eqref{eq:coordinate-comparison}}{=}\nu_i(\overline M)
 =\begin{pmatrix}\eta&*\\0&\eta^{-1}\end{pmatrix},\\
 P_{w,e}^{-1}F_i(w,Q)
 &\underset{\eqref{eq:coordinate-comparison}}{=}\nu_i\!\left(\overline{T_i^{-1}MT_i}\right)
 =\begin{pmatrix}\eta&0\\ *&\eta^{-1}\end{pmatrix},
 \end{aligned}
\end{equation}
where $\eta=\bar a/\sqrt{\bar a\bar d}$.
The matrices $P_{v,e},P_{w,e}$ and $B_e^{(j)}$ depend only on $e$
and the fixed choices, while $\eta$ and the starred entries may vary
with $Q$. Equations~\eqref{eq:transverse} and~\eqref{eq:edge-bases}
prove \eqref{eq:compatibility} on every edge, using the same matrices
\eqref{eq:frames}. This proves~(4).
\end{proof}
% \paragraph{Computing the matrices from finite data.}
% Once the labeled complex is given, $B_e^{(j)}$ can be recovered from
% three opposite-edge queries. Write $\infty=[1:0]$, $0=[0:1]$, and
% $1=[1:1]$. Choose nonzero vectors $u_\infty,u_0,u_1$
% representing $\beta_e^{(j)}(\infty),\beta_e^{(j)}(0),\beta_e^{(j)}(1)$,
% respectively, and solve $u_1=\lambda u_\infty+\mu u_0$.
% The three image lines are distinct, so $\lambda,\mu\ne0$.
% Then $B_e^{(j)}=\nu_j((\lambda u_\infty\;\;\mu u_0))$.
% Indeed, a projective transformation is determined by its action on three
% distinct points. This procedure uses the projectivity proved in~(3);
% an arbitrary permutation of the labels would not admit such a matrix.

% Once the matrices $B$ and the paired reference matrices $P$ have been
% chosen as above, one can also compute $F_i$ by a traversal of $Q$.
% Choose one direction-$i$ edge $e_0:v_0\to w_0$, oriented from type zero
% to type one in factor $i$, and set
% $F_i(v_0,Q)=P_{v_0,e_0}$ and $F_i(w_0,Q)=P_{w_0,e_0}$.
% On each of the two facets with fixed $i$th type, propagate across a
% directed edge $f:x\to y$ by $F_i(y,Q)=B_f^{(i)}F_i(x,Q)$.
% Within a lifted facet the $i$th tree coordinate is constant, so the
% comparison matrices telescope and the answer is independent of the path.
% This is the construction above with the cube's factor-$i$ basis chosen
% to be the reference basis of $e_0$. It requires no separate choice of a
% matrix at each vertex and no further compatibility tests inside the cube.

\subsection{Expansion of the cubical complex}
Finally, we prove the expansion properties of $X$
described in Lemma~\ref{lem:directional-ramanujan}.
RSV establish spectral expansion
for their complexes in odd characteristic
\cite[Theorem~6.12]{RSV19}. For completeness we provide the characteristic-two argument needed for our
construction using Morgenstern's spectral bounds on quotients of trees~\cite{Mor94}.

We write $B_i$ for the $(s_i+1)$-regular graph on $G/H$
with edges $gH\sim gdH$ for $dH\in D_i/H$. We prove the bound by identifying $B_i$ with a quotient
of $\cT_i$ and applying Morgenstern's theorem~\cite{Mor94}
to a graph cover $p:\widetilde B_i\to B_i$.
The covering map allows us to transfer the spectral bound to $B_i$.

% \begin{lemma}[Directional spectral expansion]
% For our cubical complex $X$, every $X_{i,\type}$ with $\type_i=0$
% is connected, and
% \begin{equation}\label{eq:directional-ramanujan}
%  \lambda_i(X)\le\frac{2\sqrt{s_i}}{s_i+1}
%  \qquad (1\le i\le4).
% \end{equation}
% \end{lemma}

\begin{proof}[Proof of Lemma~\ref{lem:directional-ramanujan}]
We first show that $B_i$ is connected. We fix $i$ and let $n_i=s_i+1$.
We recall the parts of the construction in \cite{RSV19} needed below. The complex $X=\Gamma\backslash\cT$ is a quotient of the Cartesian product $\cT=\cT_1\times\cdots\times\cT_4$ of four infinite trees. We fix a base vertex $o=(o_1, o_2, o_3,o_4)\in\cT$, with $o_j\in\cT_j$ to be the origin. For each $i$, we write $\Lambda_i\le\Lambda$ for the subgroup of elements whose action on $\cT_j$ fixes $o_j$ for every $j\ne i$. We write $\theta_i:\Lambda_i\to G$ for the finite quotient
map from the construction, restricted to $\Lambda_i$.
(This map reduces coefficients modulo an irreducible polynomial
of degree $r$, which maps into the residue field $\FF_{2^r}$.) The image of $\theta_i$ is a projective $2\times2$
matrix over $\FF_{2^r}$. The group $\Lambda_i$ acts transitively on $\cT_i$,
with stabilizer $H_0$ satisfying $\theta_i(H_0)=H$.
The strong approximation argument in \cite[Remark~6.10]{RSV19}
shows that $\theta_i$ remains surjective when restricted to elements
preserving the bipartition of $\cT_i$.
With $\Delta_i=\ker\theta_i$, we note that
$B_i=\Delta_i\backslash\cT_i$, and so $B_i$ is connected.

Morgenstern's result~\cite[Theorem~3.2]{Mor94} shows that
the degree-$(s_i+1)$ tree quotients in his construction have normalized adjacency eigenvalues
$|\mu|\le 2\sqrt{s_i}/(s_i+1)$ whenever $\mu\ne\pm1$.
We will construct a cover of $B_i$ in this family and
transfer its spectral bound to $B_i$.

To be able to apply Morgenstern's result, we must impose one
additional restriction on $\Delta_i$. We recall that our
quaternion coefficients belong to $\FF_2(t)$, the field
of rational functions in $t$. The restriction requires
each group element to admit a representative $x$ such
that the coefficients of both $x$ and $x^{-1}$, written
in lowest terms, have denominators not divisible by $t$
\cite[\S3]{Mor94}.

We use the auxiliary permutation action
$\rho:\Lambda\to S_3$ from \cite[\S2.6]{RSV19}
to describe this restriction. We define $\Delta_i^0
=\{\delta\in\Delta_i:\rho(\delta)\text{ is even}\}$. Since permutation parity is a homomorphism,
$\Delta_i^0$ is a subgroup of index at most two. We define $\widetilde B_i=\Delta_i^0\backslash\cT_i$. The construction ensures that the natural projection $p:\widetilde B_i\to B_i$ is a graph cover of degree $[\Delta_i:\Delta_i^0]\le2$. We now verify that $\Delta_i^0$ consists exactly of the elements admitting the representative required by Morgenstern's construction.

To that end, the local description in \cite[\S2.6]{RSV19} gives a quaternion
element $u$ with $u^2=t$. Every nonzero representative can be
written as $x=u^m y$, where $m\in\mathbb Z$ and the coefficients
of both $y$ and $y^{-1}$ have denominators not divisible by $t$.
Under $\rho$, the element $u$ acts by a transposition, while
$y$ acts by an even permutation. It follows that $\rho([x])$
is even exactly when $m$ is even. If $m=2a$, then
$x=t^a y$, so the scalar rescaling $t^{-a}x=y$ gives the required
representative. Conversely, every representative satisfying
the requirement acts by an even permutation.

Combining the above and applying \cite[Theorem 3.2]{Mor94}, we get 
\[
 \lambda_i(X) = \bigl\|R_i|_{\ell^2_0(W)}\bigr\|
 \le \frac{2\sqrt{s_i}}{s_i+1}.\qedhere
\]
\end{proof}

\bibliographystyle{alpha}
\bibliography{ref.bib}

%\appendix 

%\input{appendix}
\end{document}